\UseRawInputEncoding

\documentclass[12pt]{article}

\usepackage{authordate1-4}

\usepackage{amssymb}
\usepackage{amsthm}
\usepackage{amsmath}
\usepackage{amscd}

\newtheorem{lemma}{\bf Lemma}[section]
\newtheorem{theorem}[lemma]{\bf Theorem}
\newtheorem{claim}[lemma]{\bf Claim}
\newtheorem{corollary}[lemma]{\bf Corollary}
\newtheorem{proposition}[lemma]{\bf Proposition}

\newtheorem{definition}[lemma]{\bf Definition}

\newcommand{\eps}{\epsilon}

\begin{document}
\pagestyle{headings}

\title{\textbf{Commutative Information Algebras: \newline Representation and Duality Theory }}

\author{Juerg Kohlas \\
\small Department of Informatics DIUF \\ 
\small University of Fribourg \\ 
\small CH -- 1700 Fribourg (Switzerland) \\ 
\small E-mail: \texttt{juerg.kohlas@unifr.ch} \\
\small \texttt{http://diuf.unifr.ch/drupal/tns/juerg\_kohlas} \\
 Juerg Schmid \\
\small Institute of Mathematics \\ 
\small University of Bern \\ 
\small CH -- 3012 Bern (Switzerland) \\
\small E-mail: \texttt{juerg.schmid@math.unibe.ch} \\
}

%\author{Juerg Kohlas \\
%\small Department of Informatics DIUF \\
%\small University of Fribourg \\
%\small CH -- 1700 Fribourg (Switzerland) \\
%\small E-mail: \texttt{juerg.kohlas@unifr.ch} \\
%\small \texttt{http://diuf.unifr.ch/drupal/tns/juerg\_kohlas}}

%\author{Juerg Schmid \\
%\small Institute of Mathematics \\
%\small University of Bern \\
%\small CH -- 3012 Bern (Switzerland) \\
%\small E-mail: \texttt{juerg.schmid@math.unibe.ch}}

\date{\today}

\maketitle

%%%%%%%%%%%%%%%%%%%%%%%%%%%%%%%%%%%%%%%%%%%%%%%%%%%%%%%%%%%%%%%%%%%%%%%%%%%

\begin{abstract}
Information algebras arise from the idea that information comes in pieces which can be aggregated or combined into new pieces, that information refers to questions and that from any piece of information, the part relevant to a given question can be extracted. This leads to a certain type of algebraic structures, basically semilattices endowed with with additional unary operations. These operations essentially are (dual) existential quantifiers on the underlying semilattice. The archetypical instances of such algebras are semilattices of subsets of some universe, together with the saturation operators associated with a family of equivalence relations on this universe. Such algebras will be called {\em set algebras} in our context. Our first result is a basic representation theorem: Every abstract information algebra is isomorphic to a set algebra. When it comes to combine pieces of information, the idea to model the logical connectives {\em and}, {\em or} or {\em not} is quite natural. Accordingly, we are especially interested in information algebras where the underlying semilattice is a lattice, typically distributive or even Boolean. A major part of this paper is therefore devoted to developing explicitly a full-fledged natural duality theory - in the sense of \cite{clarkdavey98} - extending Stone resp. Priestley duality in a suitable way in order to take into account the additional operations.
\end{abstract}

\newpage

%%%%%%%%%%%%%%%%%%%%%%%%%%%%%%%%%%%%%%%%%%%%%%%%%%%%%%%%%%%%%%%%%%%%%%%%%%%

\tableofcontents

%%%%%%%%%%%%%%%%%%%%%%%%%%%%%%%%%%%%%%%%%%%%%%%%%%%%%%%%%%%%

%%%%%%%%%%%%%%%%%%%%%%%%%%%%%%%%%%%%%%%%%%%%%%%%%%%%%%%%%%%%%%%%%%

\section{Introduction and Overview}

%%%%%%%%%%%%%%%%%%%%%%%%%%%%%%%%%%%%%%%%%%%%%%%%%%%%%%%%%%%%%%%%%%

Information algebras arise from the idea that information comes in pieces which can be aggregated or combined into new pieces, that information refers to questions and that from any piece of information, the part relevant to a given question can be extracted. This view leads to two different, but essentially equivalent types of algebraic structures, \textit{domain-free} and \textit{labeled} information algebras \cite{kohlas03,kohlasschmid14}. Archetypical instances of such algebras are so-called \textit{set algebras} (for the domain-free version), resp. \textit{relational algebras} connected to relational database theory (for the labeled version). In both instances questions are represented by the sets of all their possible answers, and pieces of information are thought of as certain sets of possible answers, giving a precise meaning to the elements of information algebras. This paper will deal with domain-free type of information information algebras exclusively.

The natural question is therefore whether and to what extent abstract information algebras are isomorphic to such set algebras. Partial answers were given in \cite{kohlas03}. Here, we want to address the problem more systematically. The problem is similar to representation problems in lattice theory, where Boolean algebras or distributive lattices are shown to be isomorphic to subset algebras resp. lattices  of certain topological spaces \cite{daveypriestley02}. A substantial part of this paper is motivated by the classical duality theories for Boolean algebras resp. distributive lattices, and we extend Stone resp. Priestley duality to domain-free information algebras.

Commutative domain-free information algebras are introduced in Section \ref{sec:InfAlg}. The notion of a set (information) algebra is defined, and a few illustrative examples of such algebras are given. For a more complete presentation of information algebras we refer to \cite{kohlas03} and for more examples to \cite{poulykohlas11}. A parallel representation theory for the associated  so-called labeled information algebras must be postponed; for some partial results see \cite{kohlas03}. An information algebra induces a partial order on its elements, reflecting the information contents of the pieces of information. We show that so-called {\em truncated up-sets} (relative to this order) of elements of an information algebra may be used to construct a set algebra isomorphic to the given algebra, providing a first general representation theorem.

To see set algebras at work, we consider in Section \ref{sec:AtomAlg} atomic or atomistic information algebras. Such algebras have a very natural representation as set algebras consisting of sets of atoms, loosely speaking, of maximally informative pieces of information. This representation could be used directly to develop a representation theory for information algebras based on a Boolean algebra, since maximal ideals in such algebras are atoms in the ideal completion of the underlying Boolean algebra, resulting in an extension of Stone's representation theory for Boolean algebras.  We do not elaborate this approach, since the Boolean case is subsumed in that of information algebras based on distributive lattices, to be considered in full generality in the following Section \ref{sec:DistrLattAlg}.

The treatment of quantifiers on distributive lattices by \cite{cignoli91}, generalizing Halmos' theory of monadic Boolean algebras, will provide  the basis for a representation theory of information algebras based on distributive lattices, extending and generalizing the Boolean case. We will show in Section \ref{sec:DistrLattAlg} that, in fact, there is a full-fledged natural duality in the sense of \cite{clarkdavey98} between the categories of commutative domain-free information algebras based on distributive lattices with morphisms as defined in Subsection \ref{Homomorphisms and Subalgebras} on one side and Priestley spaces equipped with a semigroup of commuting and separating equivalences and morphisms as defined in Subsection \ref{$Q$-Priestley spaces} on the other. In Subsection \ref{Boolean CDFs} we consider the  special case of information algebras based on Bollean lattices. Finally, in Subsection \ref{Finite Distr CDFs} we look in some detail at information algebras carried by finite distributive lattices. It turns out that this class is as close to an elementary class in the sense of first order logic as one can possibly get.

For the sake of completeness, it should be noted that embedding information algebras into set algebras is not the only way to model information algebras with sets. Already the ideal completion of an information algebra embeds the information algebra into an algebra of sets, namely the algebra of its ideals. But this is \textit{not} a set algebra in the strict sense used in this paper. Also, it is well known that information algebras are closely related to information systems (in the sense of domain theory), see \cite{kohlas03}. Again, this yields not a representation theory in the sense considered here. As mentioned, most of the results contained in this paper should have, in some way or another, a counterpart in the labeled version of information algebras. This is a subject still to be worked out.

From the point of view of universal algebra, information algebras as considered in this paper can be seen as semilattices endowed with a family of (dual) existential quantifiers which form a commutative, idempotent semigroup with respect to composition. Many examples of such  structures can be found in algebraic logic, but usually related to Boolean algebras instead of semilattices, e.g.  quantifier algebras \cite{halmos62,plotkin94} or cylindric algebras \cite{HMT71,plotkin94}. As far as representation theory is concerned, there are well known and well developed theories for monadic Boolean algebras \cite{halmos62}, for cylindric algebras \cite{HMT71} and for distributive lattices with a quantifier \cite{cignoli91}.  Otherwise, to the best of our knowledge, not much is known about representation of semilattices with quantifiers, except for a few rudimentary results contained in \cite{kohlas03}.

As a notational convention, in order to improve readability, we admit writing $fx$ instead of $f(x)$ whenever it is clear from the context that $f$ is a function and $x$ a member of the domain of $f$.

%%%%%%%%%%%%%%%%%%%%%%%%%%%%%%%%%%%%%%%%%%%%%%%%%%%%%%%%%%%%

\section{Domain-free Information Algebras} \label{sec:InfAlg}

%%%%%%%%%%%%%%%%%%%%%%%%%%%%%%%%%%%%%%%%%%%%%%%%%%%%%%%%%%%%%%%%%%

%%%%%%%%%%%%%%%%%%%%%%%%%%%%%%%%%%%%%%%%%%%%%%%%%%%%%%%%%%%%%%%%%%

\subsection{Structures} \label{subsec:DoFreeInfAlg}

%%%%%%%%%%%%%%%%%%%%%%%%%%%%%%%%%%%%%%%%%%%%%%%%%%%%%%%%%%%%%%%%%%%

We will define a type of algebra describing the interaction
between ``pieces of information'' and ``questions'' as discussed
in the introduction.

\bigskip

{\em Defining operations}

\smallskip

Beginning with ``pieces of information'', let $\Phi$ be an
abstract set whose elements are thought to represent such pieces,
denoted by lower case Greek letters. We assume that $\Phi$ is
equipped with a binary operation $\cdot$ :

\begin{center}
\emph{Combination}:$\hspace{3mm} \cdot: \Phi\times\Phi
\longrightarrow \Phi$.
\end{center}

For $\phi,\psi\in\Phi$, the element $\phi\cdot\psi$ represents the
aggregation of the pieces of information represented by $\phi$
resp. $\psi$. Mimicking the intuitive properties of
``aggregation'', combination is assumed to be associative,
commutative and idempotent.  Additionally, we assume that there
exist in $\Phi$ a unit or neutral element $1$ and a null element
$0$ satisfying $1\cdot\phi = \phi\cdot 1 = \phi$ resp. $0\cdot\phi
= \phi\cdot 0 = 0$. $1$ represents vacuous information which does
not change any other information under combination. $0$ represents
contradiction and destroys any other information. Summing up,
$(\Phi;\hspace{0.5mm} \cdot, 1,0)$ is a commutative idempotent
semigroup with a neutral resp. null element.

\smallskip

Turning to ``questions'', we think of an abstract set $Q$ whose
elements represent such questions. Elements of $Q$ will typically
be denoted by $x,y,z,\ldots$ etc.. In view of the discussion in
the introduction, we will not deal with the questions $x\in Q$
themselves, but represent them, for each $x\in Q$, by a unary
operation $\epsilon_x:\Phi\longrightarrow\Phi$ which extracts,
from every $\phi\in\Phi$, the piece of information
$\epsilon_x(\phi)$ which is relevant to question $x$ (this amounts
to replacing $x$ by the graph of the map $\epsilon_x$):

\begin{center}
\emph{Extraction}:$\hspace{3mm} \epsilon_x: \Phi \longrightarrow
\Phi$.
\end{center}

The set of all such operations will be denoted by
$E(\Phi,Q)$ or just $E$ if $\Phi$ and $Q$ are clear from the
context

\smallskip

The members of $E$ will be required to satisfy all $x$-$y$-instances
(for $x,y\in Q$) of the following conditions:

\begin{enumerate}

\item $\epsilon_x(0) = 0$ \hfill (N)
\item $\phi\cdot\epsilon_x(\phi) = \phi$, for all $\phi\in\Phi$ \hfill (A)
\item  $\epsilon_x(\epsilon_x(\phi)\cdot\psi) = \epsilon_x(\phi)\cdot\epsilon_x(\psi)$, for all $\phi,\psi\in\Phi$ \hfill  (Q)
\end{enumerate}

(N) says that contradiction cannot be eliminated by extraction. (A) states that information extracted from $\phi$ is contained in $\phi$.  The crucial condition is (Q) as we shall see. Operations $\eps: \Phi \longrightarrow \Phi$ satisfying (N), (A) and (Q) will be called {\em extraction operators}.

At this point, we impose an additional condition on extraction operators, defining the scope of this paper: We require that the order of successive extractions does not matter, that is,

4. $\epsilon_x(\epsilon_y(\phi)) =
\epsilon_y(\epsilon_x(\phi))$, for all $\phi\in \Phi$  and $x,y\in Q$ \hfill (C).

Structures $(\Phi,E)$ with $(\Phi;\hspace{0.5mm}\cdot, 0,1)$ a commutative idempotent semigroup with null and unit and $E$ satisfying conditions 1. to 4. will be called {\em commutative domain-free information algebras}.

\begin{lemma}
Let $(\Phi,E)$ a commutative domain-free information algebra. Then \newline
5. $\epsilon_x(\epsilon_x(\phi)) = \epsilon_x(\phi)$, for all
$\phi\in\Phi$ and $x\in Q$. \hfill {\rm(I)}
\end{lemma}

\begin{proof}
Note that $\epsilon_x(1) = 1 \cdot \epsilon_x(1) = 1$ by (A). Using (Q), we get
$\epsilon_x(\epsilon_x(\phi)) =
\epsilon_x(\epsilon_x(\phi)\cdot 1) =
\epsilon_x(\phi)\cdot\epsilon_x(1)= \epsilon_x(\phi)\cdot 1 =
\epsilon_x(\phi)$.
\end{proof}

Note that (I) already follows from (A) and (Q).

\begin{lemma}  \label{preservation of N,A,Q}
If $E$ satisfies all instances of (N), (A), (C) and (Q), then so
does $E\cup\{\epsilon_x\circ\epsilon_y\}$, for any
$\epsilon_x,\epsilon_y\in E$.
\end{lemma}

\begin{proof}
(N) is obvious. For (A),
\begin{eqnarray*}
\lefteqn{\phi\cdot\epsilon_x \circ \epsilon_y(\phi)} \\
&& =   (\phi\cdot\epsilon_x(\phi))\cdot\epsilon_x \circ \epsilon_y(\phi) \qquad \text{by (A)} \nonumber\\
&& = \phi\cdot(\epsilon_x(\phi)\cdot\epsilon_y \circ \epsilon_x(\phi))  \qquad  \text{by (C)} \nonumber \\
&& =  \phi\cdot\epsilon_x(\phi)  \qquad \text{by (A)} \nonumber\\
&& =  \phi  \qquad  \text{by (A)} \nonumber\
\end{eqnarray*}

For (Q),
\begin{eqnarray*}
\lefteqn{\epsilon_x \circ \epsilon_y(\epsilon_x \circ \epsilon_y(\phi)\cdot\psi)} \\
&& =
\epsilon_y \circ \epsilon_x(\epsilon_x \circ \epsilon_y(\phi)\cdot\psi)  \quad
\text{ by (C)} \\
&& =
     \epsilon_y(\epsilon_x \circ \epsilon_y(\phi)\cdot\epsilon_x(\psi))
     \quad  \text{by (Q)}  \nonumber\\
&& =  \epsilon_y(\epsilon_y \circ \epsilon_x(\phi)\cdot\epsilon_x(\psi))
     \quad  \text{by (C)}  \nonumber\\
&& =  \epsilon_y \circ \epsilon_x(\phi)\cdot
     \epsilon_y \circ \epsilon_x(\psi)  \quad  \text{by (Q)}  \nonumber\\
&& =  \epsilon_x \circ \epsilon_y(\phi)\cdot
     \epsilon_x \circ \epsilon_y(\psi)  \quad  \text{by (C)}  \nonumber\
\end{eqnarray*}

For (C), $(\epsilon_x \circ \epsilon_y) \circ \epsilon_z =
\epsilon_z \circ (\epsilon_x \circ \epsilon_y)$ using associativity of
composition and (C) for $E$.
\end{proof}

\begin{corollary}  \label{preservation of N, A, C, Q}
If $E$ satisfies all instances of (N), (A), (C), (Q) and (I) then so
does $E^{\circ}$, the closure of $E$ under composition $\circ$.
Obviously, $(E^{\circ},\circ)$ is the least idempotent semigroup satisfying (N), (A),(C) and (Q) containing $E$
as a subset.
\end{corollary}

{ \bf In order to develop a meaningful algebraic theory of commutative information algebras avoiding partial morphisms, we assume henceforth, based on Corollary \ref{preservation of N, A, C, Q}, that $E$ is closed under composition}.

Altogether, we have set up a two-sorted algebra $\underline{A}=(\Phi,\cdot,1,0;\;E,\circ)$. Such algebras will be called \textit{commutative domain-free information algebras} or, for short, {\em CDF information algebras}. However, in order to avoid cluttering the paper with a plethora of CDF's, we agree that

{\bf "information algebra" will mean "commutative domain-free information algebra" if not explicitly stated otherwise}.

We introduce the following notation: Write $\underline{\Phi}$ for the commutative semigroup $(\Phi;\cdot,1,0)$, $\underline{E}$ for the commutative semigroup $(E;\circ)$ and finally $\underline{A}=(\underline{\Phi};\underline{E})$. More generally, is $S$ is any set carrying some structure, we will write $\underline{S}$ for the set equipped with the type of structure under consideration.

\bigskip

{\em Introducing order on $\Phi$}

\smallskip

It is well-known that any idempotent commutative semigroup may be equipped
with a compatible order relation in exactly two ways. Explicitly,
in $(\Phi;\hspace{0.5mm} \cdot)$ we may define an order $\leq_1$
by $\phi\leq_1\psi\ :\Longleftrightarrow\ \phi\cdot \psi = \phi$
respectively by $\phi\leq_2\psi\ :\Longleftrightarrow\ \phi\cdot
\psi = \psi$. For $\Phi$, we will use $\leq_2$:

\begin{definition}   \label{DefInfOrder}
{\bf Information order:} For $\phi,\psi\in \Phi$, we put  $\phi\leq\psi$ iff $\phi\cdot
\psi = \psi$.
\end{definition}

This is appropiate since
$\phi\cdot\psi = \psi$, in a natural way, means that $\phi$ is
less informative than $\psi$. It is easy to check that in the
ordered set $(\Phi;\leq)$ the combination $\phi\cdot\psi$ is in
fact the {\em supremum} of $\phi$ and $\psi$, that is,=
$\phi\cdot\psi = sup_{\leq_2}\{\phi,\psi\}$.
We use $\leq$ to define a
binary operation $\vee:\Phi\times\Phi\longrightarrow\Phi$ (called
join) on $\Phi$ by
\begin{center}
$\phi\vee\psi := sup_{\leq}\{\phi,\psi\} \hspace{5mm} (=
\phi\cdot\psi).$
\end{center}
This turns $\Phi$ into a join-semilattice $\underline{\Phi} = (\Phi;\vee, 1,0)$ with
least element $1$ and greatest element $0$, which neatly reflects
the fact that contradiction 0 dominates every piece of
information, and that the vacuous information 1 is contained in
every piece of information.

The interplay between the {\em alter egos} of $\underline{\Phi}$ as {\em
semigroup} resp. {\em semilattice} resp. {\em ordered set} will prove to be
very fruitful. So combination will be denoted by both $\cdot$ and $\vee$ in order
to indicate which aspect is prevalent in a given context.

Using $\leq$ and $\vee$, the conditions (N), (A) and $Q$ for an
extraction operator may be rewritten as follows:
\begin{enumerate}
\item $\epsilon(0) = 0$, \item $\epsilon(\phi) \leq \phi$ for all
$\phi \in \Phi$, \item $\epsilon(\epsilon(\phi) \vee \psi) =
\epsilon(\phi) \vee \epsilon(\psi)$, for all $\phi,\psi \in�\Phi$.
\end{enumerate}

An operator $\epsilon$ on a semilattice $(\Phi;\wedge, 0)$
satisfying these three conditions is called an {\em  existential
quantifier} in algebraic logic. However, it must be noted that in
the relevant literature rather the order relation $\leq_1$
is used to define an existential quantifier. This gives rise to a
meet-semilattice $(\Phi;\wedge,0)$ with $ \phi\wedge\psi :=
inf_{\leq_1}\{\phi,\psi\}$ and least element $0$. The three
conditions then read $\epsilon(0) = 0$, $\epsilon(\phi) \geq \phi$
and $\epsilon(\epsilon(\phi)\wedge\psi)=
\epsilon(\phi)\wedge\epsilon(\psi)$. We could have called our
variant a ``dual existential quantifier'' with the risk of
cluttering the paper with a plethora of ``duals'' -- from which we
shrank back. In any case, our choice of $\leq_2$ over $\leq_1$ is
amply justified by the natural order between pieces of
information.

\begin{lemma}  \label{extraction preserves order}
An extraction operator preserves (information) order.
\end{lemma}

\begin{proof}
Assume $\phi\leq\psi$. Since $\epsilon(\phi)\leq\phi$, we have
$\epsilon(\phi)\leq\psi$, thus $\epsilon(\phi)\cdot\psi = \psi$
and $\epsilon(\epsilon(\phi)\cdot\psi) = \epsilon(\psi)$. Using
(Q), we obtain $\epsilon(\phi)\cdot\epsilon(\psi) =
\epsilon(\psi)$, that is, $\epsilon(\phi)\leq \epsilon(\psi)$.
\end{proof}

%%%%%%%%%%%%%%%%%%%%%%%%%%%%%%%%%%%%%%%%%%%%%%%%%%%%%%%%%%%%%%%%%%%

\subsection{Homomorphisms and Subalgebras}   \label{Homomorphisms and Subalgebras}

%%%%%%%%%%%%%%%%%%%%%%%%%%%%%%%%%%%%%%%%%%%%%%%%%%%%%%%%%%%%%%%%%%%%%

Let $\underline{A} = (\underline{\Phi};\underline{E})$ and $\underline{B} = (\underline{\Psi};\underline{D})$ any two  information algebras. We do not notationally distinguish between the operations in the two algebras as their meaning will be clear from the context.

\begin{definition}    \label{CDF homos}
A pair $(f,g)$ of maps $f : \Phi \rightarrow \Psi$, $g : E \rightarrow D$ is a  homomorphism from $\underline{A}$ to $\underline{B}$ iff
\begin{enumerate}
\item $f(\phi \cdot \psi) = f(\phi) \cdot f(\psi)$ for all $\phi,\psi \in \Phi$,
\item $f(0) = 0$ and $f(1) = 1$.
\item $g(\epsilon \circ \eta) = g(\epsilon) \circ g(\eta)$ for all $\epsilon,\eta \in E$.
\item $f(\epsilon(\phi)) = g(\epsilon)(f(\phi))$ for all $\phi \in \Phi$ and $\epsilon \in E$.
\end{enumerate}
\end{definition}

Note that in a homomorphism $(f,g)$ the map $f$ is \textit{order-preserving}.

\begin{lemma}   \label{alg isos 1}
If for a homomorphism $(f,g)$ both $f^{-1}$ and $g^{-1}$ exist, then $(f^{-1},g^{-1})$ also satisfies condition 2.6.(4).
\end{lemma}

\begin{proof}
Let $(\psi,\delta)\in(\Psi,D)$. Then $\psi = f(\phi)$ and $\delta = g(\eps)$ for an unique pair $(\phi,\eps)\in(\Phi,E)$. Now $f(\eps(\phi)) = g(\eps)(f(\phi))$ by assumption. Applying $f^{-1}$ on both sides, we obtain $\eps(\phi)=f^{-1}(g(\eps)(f(\phi)))$
or $g^{-1}(\delta)(f^{-1}(\psi)) = f^{-1}(\delta(\psi))$.
\end{proof}

\begin{corollary}   \label{alg isos 2}
A homomorphism $(f,g)$ is an isomorphism iff both $f$ and $g$ are bijective.
\end{corollary}

If the maps $f$ and $g$ both are \textit{one-to-one}, then $\underline{A}$ is said to be \textit{embedded} into $\underline{B}$. If $\Phi\subseteq\Psi$ and $E\subseteq D$ are such that
  \begin{itemize}
    \item[(i)] $\Phi$ is closed under the combination operation of $\Psi$ and contains the neutral and null elements of $\Psi$,
    \item [(ii)] $E$ is closed under the composition operation of $D$, and
    \item  [(iii)] for all $\eta \in E$ and $\phi \in \Phi$, the element $\eta(\phi)$ belongs to $\Phi$,
  \end{itemize}
  then $\underline{A}$ is called a \textit{subalgebra} of $\underline{B}$. Clearly, then, the pair of the identity maps of $\Phi$ and $E$ into $\Psi$ resp. $D$  is an embedding of $\underline{A}$ into $\underline{B}$. Also, if $(f,g)$ is a homomorphism from $\underline{A}$ into $\underline{B}$, then the image $(\underline{f(\Phi)}; \underline{g(E)})$ of $\underline{A}$ is a subalgebra of $\underline{B}$.

As an example and for further reference, consider an arbitrary but fixed $\eta\in D$ and let $\eta\Psi = \{\eta\psi: \psi\in\Psi\}$. We have $\eta1 = 1$, $\eta0 = 0$ and $\eta\psi\cdot\eta\psi' = \eta(\eta\psi\cdot\eta\psi')$ by (Q), so $\eta\Psi$ is closed under the operations of $\underline{\Psi}$, making it a substructure $\eta\underline{\Psi}$ of $\underline{\Psi}$.

\begin{lemma}  \label{extraction subalgebra}
 $(\eta\underline{\Psi};\underline{D})$ is a subalgebra of $(\underline{\Psi};\underline{D})$.
\end{lemma}

\begin{proof}
(i) above is satisfied as just shown, (ii) is vacuously true, and for (iii) observe that $\eta'(\eta\psi) = \eta(\eta'\psi)\in\eta\Psi$ by (C) for any $\eta'\in D$.
\end{proof}

In the sequel we are particularly interested in homomorphisms, embeddings and isomorphisms between an arbitrary  information algebra $\underline{A}$ and so-called \textit{set algebras}, to be defined in the following section. Such embeddings and isomorphisms will be called \textit{representations} of $\underline{A}$.

Finally, we remark that from a category-theoretic point of view other types of morphisms may be more appropriate, see e.g. \cite{kohlasschmid14} for Cartesian-closed categories of information algebras.

%%%%%%%%%%%%%%%%%%%%%%%%%%%%%%%%%%%%%%%%%%%%%%%%%%%%%%%%%%%%

\subsection{Set Algebras} \label{subsec:SetAlgs}

%%%%%%%%%%%%%%%%%%%%%%%%%%%%%%%%%%%%%%%%%%%%%%%%%%%%%%%%%%%%%%%

So far, the set $\Phi$ of pieces of information as well as the set $Q$ of questions have been arbitrary abstract sets, subject only to the conditions specified for composition and extraction. We will now define a special type of  information algebras - to be called {\em set algebras} - where the elements of these sets have an internal structure, described by set-theoretical constructs over some base set $U$ ($U\neq\emptyset$). The power set of $U$ will be denoted by $P(U)$.

We may equip $P(U)$ with a lattice structure in the obvious way. To be precise, let $\underline{P}(U) := (P(U); \cap, \cup, \emptyset, U)$ the bounded distributive lattice with carrier $P(U)$, set intersection as meet, set union as join and with $\emptyset$ resp. $U$ as least rep. greatest elements. Due to our use of information order, we will mostly be concerned with the order dual $\underline{P}^d(U)$ of $\underline{P}(U)$, and especially with $(\cap,U,\emptyset)$-reducts of the latter.

The basic idea is to consider $U$ as a set of possible worlds. Questions $x\in Q$ will then be modelled by equivalence relations $\equiv_x$ on $U$, the idea being that for $u,u'\in U$ we have $u\equiv_x u'$ iff question $x$ has the same answer in the worlds $u$ resp.\!\! $u'$.

\bigskip

{\em Equivalences and saturation operators}

\smallskip

It is useful for our purposes to examine, in some detail, the set $Eq(U)$ of all equivalence relations on $U$. Recall that any equivalence $\Theta\in Eq(U)$ has an {\em alter ego} as a partition of $U$ into pairwise disjoint nonempty sets, the  {\it equivalence classes} of $\Theta$ or, for short,  $\Theta$-blocks. A $\Theta$-block thus contains, with any $u\in B$, all $u'\in U$ satisfying $(u,u')\in\Theta$. Abusing notation to the limit, we also write $B\in\Theta$ to indicate that $B$ is a $\Theta$-block, and $u\Theta u'$ instead of $(u,u')\in\Theta$.

To every $\Theta\in Eq(U)$ we associate a {\it saturation operator} $\sigma_{\Theta}:P(U)\longrightarrow P(U)$ defined by $\sigma_{\Theta}(X) = \bigcup\{B: B\in\Theta \textit{ and } B\cap X\neq\emptyset\}$, for any subset $X\subseteq U$. Accordingly, a set $X\subseteq U$ will be called $\sigma_\Theta$-{\em saturated} iff $\sigma_\Theta(X) = X$.

The following properties of saturation operators will be crucial for our purposes:

\begin{lemma}   \label{saturation operators}
Let $\Theta\in Eq(U)$ with associated saturation operator $\sigma_{\Theta}$. Then for all $X,Y\subseteq U$:

1. $\sigma_{\Theta}(\emptyset) =  \emptyset$,

2. $X\subseteq\sigma_{\Theta}(X)$,

3. $X\subseteq Y$ implies $\sigma_{\Theta}(X)\subseteq\sigma_{\Theta}(Y)$,

4. $X =\sigma_{\Theta}(X)$ and $Y =\sigma_{\Theta}(Y)$ jointly imply $X\cap Y =\sigma_{\Theta}(X\cap Y)$,

5. $\sigma_{\Theta}(\sigma_{\Theta}(X)\cap Y) = \sigma_{\Theta}(X)\cap\sigma_{\Theta}(Y)$.

6. $\sigma_{\Theta}(X\cup Y) = \sigma_{\Theta}(X)\cup\sigma_{\Theta}(Y)$.
\end{lemma}

\begin{proof}
For 1., we have  $\sigma_{\Theta}(\emptyset) = \bigcup\{B\in \Theta: B\cap\emptyset\neq\emptyset\} = \emptyset$.

 Items 2. and 3. are obvious.

 For 4., observe that $X =\sigma_{\Theta}(X)$ iff $X$ is a set union of whole $\Theta$-blocks, and that for two $\Theta$-blocks $B_1$ and $B_2$, either $B_1\cap B_2 = \emptyset$ or $B_1= B_2$.

 For 5., observe that $\sigma_{\Theta}(X)\cap Y\subseteq \sigma_{\Theta}(X)\cap\sigma_{\Theta}(Y)$,  so $\sigma_{\Theta}(\sigma_{\Theta}(X)\cap Y) \subseteq \sigma_{\Theta}(\sigma_{\Theta}(X)\cap\sigma_{\Theta}(Y)) = \sigma_{\Theta}(X)\cap\sigma_{\Theta}(Y)$ by 3. and 4. For the reverse inclusion, we have $\sigma_{\Theta}(X)\cap\sigma_{\Theta}(Y) = \bigcup\{B\in \Theta : B\cap X\neq\emptyset\neq B\cap Y\}$. Obviously, for each such $B$ we have $B\cap\sigma_{\Theta}(X) = B$, so $B\cap\sigma_{\Theta}(X) \cap Y\neq \emptyset$ and $B$ participates in the union of all $B'\in P$ forming  $\sigma_{\Theta}(\sigma_{\Theta}(X)\cap Y)$. So $\sigma_{\Theta}(X)\cap\sigma_{\Theta}(Y)\subseteq\sigma_{\Theta}(\sigma_{\Theta}(X)\cap Y)$.

 Finally, 6. is immediate.
\end{proof}

\begin{corollary}   \label{saturation and extraction}
$\sigma_{\Theta}$ is an extraction operator on the $(\cap,\emptyset)$-reduct of $\underline{P}(U)^d$,  for any $\Theta\in Eq(U)$.
\end{corollary}

\begin{proof}
Items 1., 2. and 5. in Lemma \ref{saturation operators} are just conditions (N), (A) and (Q) for an extraction operator, in their semilattice version.
\end{proof}

Recall that the {\em relational product} $\star$ of two binary relations $R,S\subseteq U\times U$ is given by $R\star S = \{(u,u')\in U\times U: \textit{ there exists } v\in U \textit{ such that } uRvSu' \}$. In general, $Eq(U)$ is not closed under $\star$. In fact, we have

\begin{lemma}   \label{star product of two equivalences}
Given $\Theta,\Gamma\in Eq(U)$, their relational product $\Theta\star\Gamma$ belongs to $Eq(U)$ iff $\Theta\star\Gamma = \Gamma\star\Theta$.
\end{lemma}

\begin{proof}
Since $u\Theta u\Gamma u $ for all $u\in U$, $\Theta\star\Gamma$ is reflexive. Now $u\Theta\star\Gamma u'$ iff for some $v\in U$ we have $u\Theta v\Gamma u'$. So $\Theta\star\Gamma$ is symmetric iff $u\Theta v\Gamma u'$ implies the existence of $w\in U$ such that $u'\Theta w\Gamma u$ for all $u,u'\in U$. So $\Theta\star\Gamma\subseteq \Gamma\star\Theta$. The reverse inclusion is obtained in the same way and we have that $\Theta\star\Gamma$ is symmetric iff $\Theta\star\Gamma = \Gamma\star\Theta$. It remains to establish transitivity of $\Theta\star\Gamma$. Assume $u\Theta\star\Gamma w$ and $w\Theta\star\Gamma u'$. So there are $x,y\in U$ such that $u\Theta x\Gamma w \Theta y\Gamma u'$. So $x\Gamma\star\Theta y$ and using $\Gamma\star\Theta$ we find $w'\in U$ such that $x\Theta w' \Gamma y$. Putting all together we have $u\Theta x\Theta w' \Gamma y\Gamma u'$ and by transitivity $u\Theta w' \Gamma u'$, that is, $u\Theta\star\Gamma u'$ as desired.
\end{proof}

Equivalences $\Theta,\Gamma$ satisfying $\Theta \star \Gamma = \Gamma\star \Theta$ will be called {\em commuting}.
The following lemma collects some properties of commuting equivalences:

\begin{lemma}   \label{commuting equivalences}
Assume $\Theta,\Gamma\in Eq(U)$ commute. Then

1. $\Theta\star\Gamma$ is the least equivalence relation on $U$ containing $\Theta$ and $\Gamma$ (as subsets of $U\times U$),

2. $\sigma_{\Theta\star\Gamma} = \sigma_{\Theta}\circ\sigma_{\Gamma}$,

3. $\sigma_{\Theta}\circ\sigma_{\Gamma} = \sigma_{\Gamma}\circ\sigma_{\Theta}$,
\end{lemma}

\begin{proof}
1. Assume $u\Theta u'$. Then $u\Theta u'\Gamma u'$ and thus $u\Theta\star\Gamma u'$, so $\Theta\subseteq\Theta\star\Gamma$, and analogously $\Gamma\subseteq\Theta\star\Gamma$. Conversely, let $\Theta,\Gamma\subseteq\Lambda\in Eq(U)$ and and assume $u\Theta\star\Gamma u'$. This means that $u\Theta v\Gamma u'$ for some $v\in U$. Hence $u\Lambda v\Lambda u'$ and so $u\Lambda u'$.

2. Let $X\subseteq U$. Then $u\in \sigma_{\Theta\star\Gamma}(X)$ iff there exists $x\in X$ such that $u\Theta\star\Gamma x$. Now $u\Theta\star\Gamma x$ iff there exists $v\in U$ such that $u\Theta v \Gamma x$. But this is equivalent with $u\in \sigma_{\Theta}(\sigma_{\Gamma}(X))$.

3. By 2. since $\Theta,\Gamma$ commute.
\end{proof}

Call a subset $\mathcal{T}\subseteq Eq(U)$ $\star${\em -closed} iff $\Theta,\Gamma\in \mathcal{T}$ implies $\Theta\star\Gamma\in \mathcal{T}$.

\begin{lemma}   \label{star-closed}
A subset $\mathcal{T}\subseteq Eq(U)$ is $\star$-closed iff $  \underline{\mathcal{T}} = (\mathcal{T}; \star\vert_{\mathcal{T}})$ is a commutative idempotent semigroup.
\end{lemma}

\begin{proof}
By Lemma \ref{star product of two equivalences} $\mathcal{T}$ is $\star$-closed iff $\star\vert_{\mathcal{T}}$ is commutative. Moreover, $\star$ is associative and idempotent on the whole of $Eq(U)$.
\end{proof}

We will refer to such semigroups shortly as {\em $\star$-semigroups} in $Eq(U)$. For an arbitrary such $\star$-semigroup $\underline{\mathcal{T}}$ let $Sat(\mathcal{T}) = \{\sigma_{\Theta}: \Theta\in\mathcal{T}\}$ and put $\underline{Sat}(\mathcal{T}) = (Sat(\mathcal{T});\circ)$.

\begin{proposition}  \label{semigroup of saturation operators}
$\underline{Sat}(\mathcal{T})$ is a commutative idempotent semigroup isomorphic to $\underline{\mathcal{T}}$.
\end{proposition}

\begin{proof}
The map $\Theta\longmapsto\sigma_{\Theta}$ from  $\underline{\mathcal{T}}$ to $\underline{Sat}(\mathcal{T})$ is one-to-one and onto as $\Theta$ may be recovered from $\sigma_{\Theta}$ by $u\Theta u'$ iff $u'\in \sigma_{\Theta}(\{u\})$. It is a semigroup homomorphism by Lemma \ref{commuting equivalences}.
\end{proof}

Our interest in $\star$-semigroups and their associated semigroups of saturation operators is based on the following case:

Let $(\underline{\Phi};\underline{E})$ be any  information algebra. For $\eps\in E$ define an equivalence relation $\equiv_\epsilon$ in $Eq(\Phi)$ by $\phi \equiv_\epsilon \psi$ if $\epsilon(\phi) = \epsilon(\psi)$, that is, $\equiv_\epsilon$ is   $ker\, \epsilon$, the {\em kernel} of $\epsilon$.

\begin{theorem} \label{kernels commute}
For any $\eps,\eta\in E$, we have $ker\,\eps\: \star\: ker\,\eta = ker\,(\eps\circ \eta)$, that is, $\underline{\mathcal{E}} = (\{ker\,\eps: \eps\in E\},\star)$ is a $\star$-semigroup in $Eq(\Phi)$.
\end{theorem}

\begin{proof}
Assume first that $(\phi,\psi)\in ker\,\eps\: \star\: ker\,\eta$. So there is $\chi\in\Phi$ such that $\eps\phi = \eps\chi$ and $\eta\chi = \eta\psi$. Now $\eta\eps\phi \overset{ \text{by ass.} }{=} \eta\eps\chi \overset{ \text{(C)} }{=} \eps\eta\chi        \overset{ \text{by ass.} }{=} \eps\eta\psi \overset{ \text{(C)} }{=} \eta\eps\psi$. It follows that $(\phi,\psi)\in ker\,\eta\circ\eps = ker\,\eps\circ\eta$.

Conversely, assume that $(\phi,\psi)\in ker\,(\eps\circ\eta)$. So $\eps\eta\phi = \eps\eta\psi$ and using (C), we obtain
\begin{eqnarray} \label{instead of lozenge}
\eta\eps\phi  = \eps\eta\phi = \eps\eta\psi = \eta\eps\psi.
\end{eqnarray}
Put $\xi := \eps\phi\cdot\eta\psi$. Then $\eps\xi = \eps(\eps\phi\cdot\eta\psi) \overset{ \text{(I),(Q)} }{=} \eps\phi\cdot\eps\eta\psi \overset{(\ref{instead of lozenge})}{=} \eps\phi\cdot \eta\eps\phi \overset{ \text{(A)} }{=} \eps\phi$. Similarly, one obtains $\eta\xi = \eta\psi$. So $(\phi,\psi)\in ker\,\eps\: \star\: ker\,\eta$.
\end{proof}

\bigskip

{\em Construction of set algebras}

\smallskip

We will construct an information algebra $SetAlg(\Phi,\mathcal{T})$ based on a join-subsemilattice $\underline{\Phi}$ of $\underline{P}^d(U)$ containing $U$ and $\emptyset$,  and a $\star$-semigroup $\underline{\mathcal{T}}$ in $Eq(U)$.

This means that $\Phi$ is closed under ordinary set intersection and that the semilattice operation $\vee$ on $\Phi$ is given by $\phi\vee\psi = \phi\cap\psi = \phi\cdot\psi$, for all $\phi,\psi\subseteq U$ belonging to $\Phi$. The corresponding order $\leq$ on $\Phi$ is then $\phi\leq\psi$ iff $\phi\supseteq\psi$, and $U\leq \phi\leq \emptyset$ for all $\phi\in\Phi$.

Note that for any $\star$-semigroup $\underline{\mathcal{T}}$ in $Eq(U)$ we have $\sigma_{\Theta}(U) = U$ and $\sigma_{\Theta}(\emptyset) = \emptyset$ for all $\sigma_{\Theta}\in Sat(\mathcal{T})$. $\underline{\mathcal{T}}$ will be called $\Phi$-{\em compatible} (or just compatible, if $\Phi$ is clear from the context) iff $\Phi$ is closed under all $\sigma_{\Theta}\in Sat(\mathcal{T})$, that is, $\sigma_{\Theta}(\phi)\in \Phi$ for all $\phi\in\Phi, \Theta\in \mathcal{T}$.

Define $Sat_{\Phi}(\mathcal{T}) = \{\sigma_{\Theta}\vert_{\Phi}: \Theta\in\mathcal{T}\}$, and $\underline{Sat}_{\Phi} =(Sat_{\Phi}(\mathcal{T}),\circ)$.

\begin{theorem}  \label{def set algebras}
Let $\underline{\Phi}$ be a be a $(\cap,U,\emptyset)$-subsemilattice of $\underline{P}^d(U)$ containing $U$ and $\emptyset$, and $\underline{\mathcal{T}}$\hspace{0.5mm} a  $\Phi-compatible$ $\star$-semigroup in $Eq(U)$. Then $SetAlg(\underline{\Phi},\underline{\mathcal{T}}) := (\underline{\Phi};\underline{Sat}_{\Phi})$  is an information algebra.
\end{theorem}

\begin{proof}
Corollary \ref{saturation and extraction}, Lemma \ref{commuting equivalences} and Proposition \ref{semigroup of saturation operators}.
\end{proof}

The algebras described by Theorem \ref{def set algebras} will be called {\em set algebras} in the sequel. As it will turn out, they are the archetypes of  information algebras.

\bigskip

{\em A special type of set algebras}

\smallskip

Any set algebra  $SetAlg(\underline{\Phi},\underline{\mathcal{T}})$  is forced by definition to contain, as members of $\Phi$,  many sets which are unions of blocks of some equivalence $\Theta\in\mathcal{T}$. Is it possible to have a set algebra where $\Phi$ consists precisely of {\em all} possible unions of this type? The following proposition shows that is the case iff $\mathcal{T}$ satisfies a simple property:

\begin{theorem}   \label{SQ-algebras}
Let $\underline{\mathcal{T}}$ be a $\star$-semigroup in $Eq(U)$ and put $\Phi=\{X\subseteq U: X = \bigcup_i B_i\text{ where } B_i\in \Theta \text{ for some } \Theta\in\mathcal{T}\}$. Then $SetAlg(\underline{\Phi},\underline{\mathcal{T}})$ is a set algebra if and only if $\mathcal{T}$ is downwards directed in $Eq(U)$ ordered by standard set inclusion.
\end{theorem}

\begin{proof}
We only need to show that $\underline{\Phi}$ is a  join-subsemilattice of $\underline{P}^d(U)$ containing $U$ and $\emptyset$. Now $\emptyset\in\Phi$ as the set union of the empty collection of blocks from any $\Theta\in\mathcal{T}$, and $U\in\Phi$ as the set union of {\em all} blocks of any $\Theta\in\mathcal{T}$. It remains to show that $\Phi$ is closed under set intersection. Then $\bigcup_i B_i\cap\bigcup_jB_j' = \bigcup_{i,j}(B_i\cap B_j')$ for blocks $B_i\in \Theta$, $ B_j'\in \Theta'$, so it suffices to show that $B_i\cap B_j'$ is a union of blocks of some $\Theta''\in\mathcal{T}$ for all $i,j$. This is clearly the case exactly if for any $\Theta,\Theta'\in\mathcal{T}$ there exists $\Theta''\in\mathcal{T}$ such that $\Theta''\subseteq\Theta$ and $\Theta''\subseteq\Theta'$.
\end{proof}

\bigskip

\subsection{A general representation theorem}  \label{GenRep Thm}

\smallskip

If $(X,\leq)$ is any ordered set, $V\subseteq X$ is called an {\em up-set} if $v\in V$ and $v\leq x$ jointly imply that $x\in V$. For any $x\in V$, the {\em principal up-set} generated by $x$ is given by  $\uparrow\!x =\{y\in V: x\leq y\}$. We write $\mathcal{U}_p(X)$ for the collection of all principal up-sets in $X$, considered as an ordered set under ordinary set inclusion.

In the following, $(\underline{\Phi};\underline{E})$ will denote an arbitrary but fixed  information algebra. Put $\Phi_0 := \Phi\setminus\{0\}$. Our aim is to construct a set algebra based on the universe $\Phi_0$.

For any $\phi,\psi\in \Phi_0$, we have $\uparrow\!\phi\;\cap\uparrow\!\psi = \uparrow\!(\phi\vee \psi)$ if $\phi\vee\psi\in\Phi_0$, and $\uparrow\!\phi\;\cap\uparrow\!\psi = \emptyset$ if $\phi\vee\psi = 0$. Let $\mathcal{U}_p^+(\Phi_0) := \mathcal{U}_p(\Phi_0)\cup\{\emptyset\}$.

It follows that $\underline{\mathcal{U}}_p^+(\Phi_0) = (\mathcal{U}_p^+(\Phi_0); \cap, \Phi_0,\emptyset)$ is a ($\cap,\Phi_0,\emptyset$)-subsemilattice of $\underline{P}^d(\Phi_0)$.

While $\Phi_0$ is not be closed under the join operation of $\underline{\Phi}$ (unless $\Phi_0$ happens to contain a greatest element), it is closed under all extractions $\eps\in E$ since $\eps\phi = 0$ implies $\phi = 0$ as $\eps\phi\leq\phi$. Similarly, if $\eps\phi = \eps 0$ for some $\phi\in\Phi$ and $\eps\in E$, then $\eps\phi = 0$ and again $\phi = 0$. In other words, the $\equiv_{\eps}$-class of $0$ is $\{0\}$, where $\equiv_{\eps}$ is infix for the kernel $ker\, \eps$ of $\eps$. It follows that $\Phi_0$ is also closed under $ker\, \eps$. Abusing notation in a trivial way, we do not distinguish  between $\eps$ resp. $\equiv_{\eps}$ and their restrictions to $\Phi_0$. Let $\mathcal{E} = \{ker\,\eps: \eps\in E\}$. So $\mathcal{E}$ is a $\star$-semigroup in $Eq(\Phi_0)$ by Theorem \ref{kernels commute}. Denote the saturation operator on $\underline{P}^d(\Phi_0)$ associated with $\equiv_{\eps}$ by $\sigma_{\eps}$, and put $Sat(\mathcal{E}) := \{\sigma_{\eps}: \eps\in E\}$

\begin{lemma}
$\mathcal{U}_p(\Phi_0)$ is closed under all $\sigma_\eps\in Sat(\mathcal{E})$.
\end{lemma}

\begin{proof}
We have to show that that $\sigma_\eps(\uparrow\!\phi) \in \mathcal{U}_p(\Phi_0)$ for all $\uparrow\!\phi\in\mathcal{U}_p(\Phi_0)$. We have $\phi\equiv_{\eps}\epsilon(\phi)$ since $\epsilon(\phi) = \eps(\epsilon(\phi))$. Let $\psi\geq\epsilon(\phi)$ and consider $\chi=\phi\vee\epsilon(\psi)\in\ \uparrow\!\phi$: Using (Q), we obtain $\epsilon(\chi) = \epsilon(\phi\vee\epsilon(\psi)) = \epsilon(\phi)\vee\epsilon(\psi)$. But $\psi\geq\epsilon(\phi)$ implies $\epsilon(\psi)\geq\epsilon(\phi)$, so we get $\epsilon(\chi) = \epsilon(\psi)$, that is, $\chi\equiv_{\eps}\psi$ and thus $\psi\in\sigma_{\eps}(\uparrow\!\phi)$. Conversely, if $\chi\geq\phi$ and $\chi\equiv_{\eps}\psi$, then $\epsilon(\psi) = \epsilon(\chi) \geq \epsilon(\phi)$. Summing up, we obtain
\begin{eqnarray} \label{eq:HomOfExtr}
\sigma_{\eps}(\uparrow\!\phi) =\ \uparrow\!(\epsilon(\phi))
\end{eqnarray}
so indeed $\sigma_{\eps}(\uparrow\!\phi)\in\mathcal{U}_p(\Phi_0)$.
\end{proof}

Put $\underline{Sat}(\mathcal{E}) := (Sat(\mathcal{E}),\circ)$. Consequently,

\begin{lemma}
$(\underline{\mathcal{U}}_p^+(\Phi_0);\underline{Sat}(\mathcal{E}))$ is a set algebra, to be called the {\em principal up-set algebra} associated with  $(\underline{\Phi};\underline{E})$.
\end{lemma}

Consider the maps  $\verb"i":\Phi\longrightarrow\mathcal{U}_p^+(\Phi_0)$ given by $\phi\longmapsto\uparrow\!\phi$ for $\phi\in\Phi_0$ and $\verb"i" 0 =\emptyset$, resp. $\verb"j":E\longrightarrow Sat(\mathcal{E})$ given by $\eps\longmapsto\sigma_\eps$.

\begin{lemma}
The pair $(\verb"i",\verb"j")$ provides an isomorphism between $(\underline{\Phi};\underline{E})$ and $(\underline{\mathcal{U}}_p^+(\Phi_0);\underline{Sat}(\mathcal{E}))$.
\end{lemma}

\begin{proof}
Note first that both $\verb"i"$ and $\verb"j"$ obviously are one-to-one and onto. Further, $\verb"i"$ preserves $\vee$, $1$ and $0$: If $\phi\vee\psi\neq 0\in\Phi$, then $\uparrow\!(\phi\vee\psi) = \uparrow\!\phi\ \cap\uparrow\!\psi$; if $\phi\vee\psi = 0\in\Phi$ then $\uparrow\!\phi\ \cap\uparrow\!\psi = \emptyset$. Further, $\uparrow\!1 = \Phi_0$, and $\verb"i" 0 =\emptyset$ by definition.

Also, $\verb"i"$ and $\verb"j"$ satisfy Def. 2.6.(4): By (\ref{eq:HomOfExtr}) above, we have $\verb"j"(\eps)(\verb"i"\phi)= \newline \sigma_\eps(\uparrow\phi) = \uparrow\!(\eps(\phi)) = \verb"i"(\epsilon(\phi))$ for $\phi\in\Phi_0$, and $\verb"j"(\eps)(\verb"i"0) = \verb"j"(\eps)(\emptyset) = \sigma_\eps(\emptyset) = \newline = \emptyset = \verb"i"(0) = \verb"i"(\eps0)$.

Finally, $\verb"j"$ preserves $\circ$: We have to show that  $\verb"j"(\eps \circ\eta)(\uparrow\phi) =(\verb"j"(\eps) \circ \verb"j"(\eta))(\uparrow\phi)$ for all $\uparrow\!\phi\in\mathcal{U}_p(\Phi_0)$. Observe first that the least element of $\eps(\uparrow\!\eta(\phi))$ is $\eps(\eta(\phi))$ (since $\eps$ is order-preserving). Hence $\uparrow\!\eps(\uparrow\!\eta(\phi)) = \uparrow\!\eps(\eta(\phi))$, that is, $\sigma_\epsilon(\sigma_\eta(\uparrow\!\phi)) = \sigma_{\epsilon\circ\eta}(\uparrow\!\phi)$, using (\ref{eq:HomOfExtr}). The case $\phi = 0$ is trivial.
\end{proof}

\begin{theorem} \label{basic repr thm}
Every  information algebra is isomorphic to a set algebra, more precisely, to its principal up-set algebra.
\end{theorem}

The semilattice part of this result is not very surprising, since any ordered set is order-isomorphic to the collection of all its principal up-sets. The point of working with $\Phi_0$ instead of $\Phi$ is to have $\emptyset$ as the image of $0\in\Phi$ (instead of $\uparrow\!0)$. So the main content of Thm. \ref{basic repr thm} is that the extraction part of any  information algebra may be modeled by (the saturation operators of) a $\star$-semigroup of compatible equivalence relations on the underlying semilattice.

Note that using $\Phi_0$ in order to obtain a set algebra representation is not compulsory. Indeed, much of the rest of this paper is devoted to showing that using other base sets, possibly equipped with additional structure, will produce representations offering deeper insight into the properties of not only the information algebras concerned, but also of the their morphisms.

%%%%%%%%%%%%%%%%%%%%%%%%%%%%%%%%%%%%%%%%%%%%%%%%%%%%%%%%

\subsection{Examples} \label{subsec:Expl}

%%%%%%%%%%%%%%%%%%%%%%%%%%%%%%%%%%%%%%%%%%%%%%%%%%%%%%%%%

%%%%%%%%%%%%%%%%%%%%%%%%%%%%%%%%%%%%%%%%%%%%%%%%%%%
\textit{Algebra of Strings}
%%%%%%%%%%%%%%%%%%%%%%%%%%%%%%%%%%%%%%%%%%%%%%%%%%%%%%%%%

\smallskip

Consider a finite alphabet $\Sigma$, the set $\Sigma^{*}$ of finite strings over $\Sigma$, including the empty string $\epsilon$, and the set $\Sigma^{\omega}$ of infinite strings over $\Sigma$. Let $\Sigma^{**} = \Sigma^{*} \cup \Sigma^{\omega} \cup \{0\}$, where $0$ is a symbol not contained in $\Sigma$. For two strings $r,s \in \Sigma^{**}$, define $r \leq s$, if $r$ is a prefix of $s$ or if $s = 0$. The empty string is a prefix of any string. Define a combination operation $\cdot$  in $\Sigma^{**}$ as follows:

\begin{eqnarray}
r \cdot s = \left\{ \begin{array}{ll} s, & \textrm{if}\ r \leq s, \\ r & \textrm{if}\ s \leq r, \\ 0, & \textrm{otherwise}. \end{array} \right.
\nonumber
\end{eqnarray}

Clearly, $\underline{\Sigma}^{**} = (\Sigma^{**},\cdot, \varepsilon,0)$ is a commutative idempotent semigroup with $\varepsilon$ as unit element and $0$ as null element of combination. For extraction, we define operators $\epsilon_{n}$ for any $n \in \mathbb{N}$ and also for $n = \infty$. Let $\epsilon_{n}(s)$ be the prefix of length $n$ of string $s$, if the length of $s$ is at least $n$, and let $\epsilon_{n}(s) = s$ otherwise. In particular, define $\epsilon_{\infty}(s) = s$ for any string $s$ and $\epsilon_{n}(0) = 0$ for any $n$. It is easy to verify that any $\epsilon_{n}$ maps $\Sigma^{**}$ into itself, and that it satisfies conditions (N), (A), and (Q) for an extraction operator. Moreover,  $\underline{E} =( \{\epsilon_{n}: n \in \mathbb{N}�\cup \{\infty\}\},\circ)$ is a commutative and idempotent semigroup under composition $\circ$ of maps. It follows that  the so-called \textit{string algebra} $(\underline{\Sigma}^{**},\underline{E})$ is an instance of a  information algebra.

\bigskip

%%%%%%%%%%%%%%%%%%%%%%%%%%%%%%%%%%%%%%%%%%%%%%
\textit{Multivariate Algebras}
%%%%%%%%%%%%%%%%%%%%%%%%%%%%%%%%%%%%%%%%%%%

\smallskip

 In many applications a set of variables is considered and the information one is interested in concerns the values of certain groups of variables, similar to ordinary relational algebra in database theory (see \cite{kohlas03} for more general relational information algebras). So, consider a countable family of variables $ X = \{X_i: i\in \mathbb{N}\}$,  and let $V_{i}$ denote the set of possible values of the variable $X_i$. For a subset $s\subseteq X$ of variables consider
\begin{eqnarray}
V_s = \prod_{X_i \in s} V_i
\nonumber
\end{eqnarray}
as the set of possible answers relative to $s$.
Let
\begin{eqnarray}
V_{\omega} = \prod_{i=1}^{\infty} V_{i}.
\nonumber
\end{eqnarray}
and put $\Phi = P(V_{\omega})$, the powerset of $V_{\omega}$. Note that the elements of $V_{\omega}$ are the sequences $t = (t_{1},t_{2},\ldots)$ with $t_{i} \in V_{i}$. An element $\phi \in�\Phi$ may be interpreted as a piece information, which states that a generic element $t\in V_{\omega} $ belongs to the set $\phi$. Within $\Phi$ we define combination by  set intersection, which represents aggregation of information:
\begin{eqnarray}
\phi \cdot \psi = \phi \cap \psi.
\nonumber
\end{eqnarray}
Equipped with this operation, $\Phi$ becomes an idempotent commutative semigroup $\underline{\Phi}$ with least element $V_{\omega}$ and greatest element $\emptyset$ under the associated information order (given by $\psi \leq \phi$ iff $\phi \subseteq \psi$). The smaller the subset representing a piece of information about elements of $V_{\omega}$ is, the more information it contains.

Let $s$ be any subset of $X$. Define, for any sequence $t$ in $V_{\omega}$, its {\em restriction} to $ s$, denoted by  $t \vert s$, as follows: If $s = \{X_{i_1},X_{i_2},\dots\}$, then  $t \vert s = (t_{i_1},t_{i_2},\dots)$. Also, define an equivalence relation $\equiv_s$ in $V_{\omega}$ by
\begin{eqnarray}
t \equiv_{s} t' \textrm{ iff}\ t \vert s = t' \vert s.
\nonumber
\end{eqnarray}

It is easy to see that the relational product $\equiv_s\star \equiv_{s'}$is $ \equiv_{s\cap s'}$. It follows that any two of such equivalence relations commute, and thus so do their associated saturation operators. Let $\mathcal{E}= \{\equiv_s: s\subseteq X\}$, $Sat(\mathcal{E})$ be the set of all saturation operators $\sigma_s$  associated with $\equiv_s$ for $s\subseteq X$ and finally $\underline{Sat}(\mathcal{E}) = (Sat(\mathcal{E}),\circ)$.

It is immediate that $\sigma_s$ maps $\Phi$ into $\Phi$ and that $\sigma_s(\emptyset) = \emptyset$ for all $s\subseteq X$. So $(\underline{\Phi};\underline{Sat}(\mathcal{E}))$ is an information algebra - in fact a set algebra - by Theorem \ref{def set algebras}. It is commonly called  the \textit{multivariate} algebra; also, the sets $\sigma_s(\phi)$ are called \textit{cylindric} over $s$.

The multivariate algebra is an information algebra closely related to \textit{relational algebras} as used in relational database systems (see \cite{kohlas03,kohlasschmid14}). We may also consider the set $\Phi'$ consisting of all subsets of $V_{\omega}$ which are cylindric over some \textit{finite} $s\subseteq X$  (plus $V_{\omega}$) and limit ourselves to operators $\sigma_s$ for finite $s$. The resulting system $(\underline{\Phi}';\underline{Sat}(\mathcal{E})')$
is an information algebra, in fact a subalgebra of $(\underline{\Phi};\underline{Sat}(\mathcal{E}))$  (this is the case since the set intersection of sets cylindric over $r$ respectively $s$ is cylindric over $r \cup s$).

\bigskip

%%%%%%%%%%%%%%%%%%%%%%%%%%%%%%%%%%%%%%%%
\textit{Lattice-Valued Algebras}
%%%%%%%%%%%%%%%%%%%%%%%%%%%%%%%%%%%%%%%%%

\smallskip

Similar to the multivariate model consider a finite family of variables $X_{i}$, $i = 1,\ldots,n$ with variable $X_{i}$ taking values in a finite set $V_{i}$, and let $V$ be the cartesian product $V_{1} \times \cdots \times V_{n}$. Further let $\Lambda$ be a \textit{bounded distributive lattice} with greatest element $\top$ and smallest element $\bot$. Consider the set $\Phi$ of all maps $\phi : V \rightarrow \Lambda$, and define an operation of \textit{combination} $\phi \cdot \psi$ on $\Phi$ by
\begin{eqnarray}
(\phi \cdot \psi)(t) = \phi(t) \wedge \psi(t) \textrm{ for all}\ t \in V.
\nonumber
\end{eqnarray}
Obviously, this defines an idempotent semigroup $\underline{\Phi}$ with unit element $1$ given by $1(t) = \top$ for all $t \in \Lambda$ and null element $0$ by $0(t) = \bot$ for all $t$. Note that under the information order we have $\phi \leq \psi$ in $\Phi$ iff  $\phi(t) \geq \psi(t)$ for all $t \in V$.

For any subset $s$ of the index set $r = \{1,\ldots,n\}$ we introduce an operator $\epsilon_{s}$ mapping $\Phi$ into $\Phi$, defined by
\begin{eqnarray}
\epsilon_s(\phi)(t) = \bigvee\{\phi(u_1,\ldots,u_n): u_i = t_i \text{ for } i\in s \text{ and } u_i\in V_i \text{ for } i\notin s\}
\nonumber
\end{eqnarray}
Using the distributivity of the lattice $\Lambda$ it is easy to verify that all of these operators $\epsilon_{s}$ are extraction operators on $\underline{\Phi}$ . Further, the set $E = \{\epsilon_{s}:s \subseteq r\}$ is a commutative, idempotent semigroup under composition. Thus $(\underline{\Phi};\underline{E})$  is an information algebra; in fact it is a distributive information algebra as will be considered in Section \ref{sec:DistrLattAlg}. Also, it is a particular case of a semiring induced valuation algebra as considered in \cite{kohlaswilson08}.

There are many more instances  of  information algebras related to algebraic logic, graph theory,  linear algebra and convex sets, and other topics as well. We refer to \cite{kohlas03,poulykohlas11} for further examples.

\bigskip

%%%%%%%%%%%%%%%%%%%%%%%%%%%%%%%%%%%%%%%%%%%%%%%%%%%%%%%%%%%%
\textit{Ideal Completions}
%%%%%%%%%%%%%%%%%%%%%%%%%%%%%%%%%%%%%%%%%%%%%%%%%%%%%%%%%%%%

\smallskip

In an algebra $(\underline{\Phi};\underline{E})$, a {\em consistent} set of pieces of information $I$ is a nonempty
subset $I\subseteq \Phi$ such that (i) with any element $\phi \in I$ also all elements $\psi \leq \phi$ implied by $\phi$ (or contained in $\phi$) belong to $I$, and (ii) with any two elements $\phi,\psi \in I$ also their combination $\phi \cdot \psi$ belongs to $I$. Such sets are just \textit{ideals} in the context of the semilattice $\underline{\Phi}$. Ideals not equal to $\Phi$ are called proper. Consistent sets (also called {\em theories}) may also be thought of as pieces of information. In fact, we may define among them operations of combination and extraction.

 Let $I(\Phi)$ denote the family of all ideals contained in $\Phi$. We define the following two operations for ideals $I_{1},I_{2},I\in I(\Phi)$:

\begin{enumerate}
\item \textit{Combination:} $I_1\cdot I_2 = \{\phi \in \Phi:\phi \leq \phi_1 \cdot \phi_2 \textrm{ for some} \phi_1\in I_1,\phi_2\in I_2\}$,
\item \textit{Extraction:} $\hat{\epsilon}(I) = \{\phi \in \Phi:\phi \leq \epsilon(\psi) \textrm{ for some}\ \psi \in I\}$.
\end{enumerate}

Let $\hat{E} = \{\hat{\eps}; \eps\in E\}$, $\underline{\hat{E}}=(\hat{E},\circ)$ and  $\underline{I}(\Phi)  = ( I(\Phi), \cdot, \{1\}, \Phi)$. It is not hard to check that $(\underline{I}(\Phi); \underline{\hat{E}})$  is an information algebra \cite{kohlas03,kohlasschmid14}, called the {\em  ideal completion} of $(\underline{\Phi};\underline{E})$. Moreover, $(\underline{\Phi};\underline{E})$ embeds into $(\underline{I}(\Phi); \underline{\hat{E}})$
 by the pair of maps $\phi \mapsto \downarrow\!\phi$ and $\epsilon\mapsto\hat{\epsilon}$, where the {\em down-set} $\downarrow\!\phi = \{\psi: \psi\leq \phi\}$ is the principal ideal generated by $\phi$. Ideal completions play an important role for the discussion of compact information algebras, see  \cite{kohlas03}. It is well-known that $I(\Phi)$, ordered by ordinary set inclusion, is a complete lattice.

%%%%%%%%%%%%%%%%%%%%%%%%%%%%%%%%%%%%%%%%%%%%%%%%%%%%%%%%%%%%

\section{Atomic Algebras} \label{sec:AtomAlg}

%%%%%%%%%%%%%%%%%%%%%%%%%%%%%%%%%%%%%%%%%%%%%%%%%%%%%%%%%%%%%

In many examples of information algebras there exist maximally informative elements. The concept of such elements is captured by the notion of an atom.  Here is the formal definition:

\begin{definition} \label{def:DomFreeAtoms}
Let $\underline{A} = (\underline{\Phi};\underline{E})$. An element $\alpha\in\Phi$ is called an atom, if
\begin{enumerate}
\item $\alpha \not= 0$,
\item If $\phi \in \Phi$, then $\alpha \leq \phi$ implies either $\alpha = \phi$ or $\phi = 0$.
\end{enumerate}
\end{definition}

The following lemma lists some properties of atoms: \footnote{We remark that in order theory an atom usually is a minimal, not a maximal element. The present concept corresponds, in a natural way, to our use of information order.}

\begin{lemma} \label{le:PropOfDomFreeAtoms}
\begin{enumerate}
Let $\underline{A} = (\underline{\Phi};\underline{E})$.
\item If $\alpha$ is an atom and $\phi \in \Phi$, then either $\alpha \cdot \phi = \alpha$ or $\alpha \cdot \phi = 0$.
\item If $\alpha$ is an atom and $\phi \in \Phi$, then either $\phi \leq \alpha$ or $\alpha \cdot \phi = 0$.
\item If $\alpha$ and $\beta$ are atoms, then either $\alpha = \beta$ or $\alpha \cdot \beta = 0$.
\end{enumerate}
\end{lemma}

\begin{proof}
Let $\alpha$ be an atom and $\phi \in \Phi$. Then $\alpha \leq \alpha \cdot \phi$. Since $\alpha$ is an atom we have either $\alpha \cdot \phi = \alpha$ or $\alpha \cdot \phi = 0$. In the first case $\phi \leq \alpha$. This proves the first two items.

Assume $\alpha$ and $\beta$ are atoms. Then $\alpha \leq \alpha \cdot \beta$, hence either $\alpha \cdot \beta = 0$ or $\alpha = \alpha \cdot \beta$, which means that $\beta \leq \alpha$, thus $\alpha = \beta$.
\end{proof}

Let $At(\Phi)$ (or just $At\Phi$ if it improves readability)  denote the set of all atoms of $\Phi$. If $\phi \leq \alpha$, this means that $\alpha$ implies $\phi$. Let $At(\phi) = \{\alpha \in At\Phi:\phi \leq \alpha\}$ be the set of all atoms implying $\phi$. We define different types of  information algebras, depending on the occurrence of atoms:

\begin{definition} \label{def:DomFreeAtom}
Let $\underline{A} = (\underline{\Phi};\underline{E})$.
\begin{enumerate}
\item $\underline{A}$ is called atomic if $At(\phi)\neq\emptyset$ for all $0\neq\phi \in\Phi$.
\item $\underline{A}$ is called atomistic if $\phi = \inf At(\phi)$ for all $0\neq\phi\in\Phi$.
\item $\underline{A}$ is called completely atomistic, if it is atomistic and if for all $\emptyset\neq A\subseteq At\Phi$ there exists $\phi \in\Phi$ such that $A = At(\phi)$.
\end{enumerate}
\end{definition}

If $\underline{A}$ is an atomic  information algebras, we will construct an associated set algebra (see Section \ref{subsec:SetAlgs}) based on the set of atoms $At\Phi$. Recall (see Lemma \ref{extraction subalgebra}) that $\eps\underline{A} := (\eps\underline{\Phi};\underline{E})$ is a subalgebra of $\underline{A}$.

\begin{lemma}  \label{preservation of atoms}
If $\underline{A}$ is atomic and $\eps\in E$, then $\eps\underline{A}$ is atomic and $At(\eps\Phi) = \eps(At(\Phi)) = \{\eps\phi: \phi\in At(\Phi)\}$.
\end{lemma}

\begin{proof}
Let $\alpha\in At\Phi$. Then $\alpha\neq 0$ and thus $\eps\alpha\neq 0 $. Assume $\eps\alpha\leq \eps\phi$ for some $\phi\in\Phi$. Then $\eps(\alpha\cdot\eps\phi) = \eps\alpha\cdot\eps\phi = \eps\phi$. Since $\alpha$ is an atom in $\Phi$ ,we have either $\alpha\cdot\eps\phi = \alpha$ or $\alpha\cdot\eps\phi = 0$. In the first case, $\eps\phi = \eps(\alpha\cdot\eps\phi) = \eps\alpha$, in the second, $\eps\phi = \eps(\alpha\cdot\eps\phi) = \eps0 = 0$. So $\eps\phi$ is an atom in $\eps\Phi$.

Conversely, assume $0\neq\eps\phi\in\eps\Phi$. Since $\underline{A}$ is atomic, there exists $\xi\in At\Phi$ such that $\eps\phi\leq \xi$ and thus $\eps\phi\leq\eps\xi$. As shown above, $\eps\xi$ is an atom in $\eps\Phi$, so $\eps\underline{A}$ is atomic. If $\eps\phi$ is an atom in $\eps\Phi$ itself, then obviously $\eps\phi = \eps\xi$ and thus $\eps\phi\in \eps At\Phi$.
\end{proof}

Let $\equiv'_\eps$ be the restriction of $\equiv_\eps$ (that is, $ker\,\eps$) to $At\Phi$.

\begin{lemma}  \label{restricted kernels commute}
Assume $\underline{A}$ is atomic. Then $\equiv'_\eps$ and $\equiv'_\eta$ commute for $\eps,\eta\in E$.
\end{lemma}

\begin{proof}
Let $\alpha,\beta\in At\Phi$ and assume $\alpha\equiv'_\eps\star\equiv'_\eta\beta$. So there exists $\gamma\in At\Phi$ such that $\alpha\equiv'_\eps\gamma$ and $\gamma\equiv'_\eta\beta$, that is, $\eps\alpha = \eps\gamma$ and $\eta\gamma = \eta\beta$. This implies $\eta\eps\alpha = \eta\eps\gamma$ and $\eps\eta\gamma = \eps\eta\beta$. By (C), this results in $\eps\eta\alpha = \eta\eps\beta$ (*).

Consider $\eta\alpha\cdot\eps\beta$: We see that $\eta(\eta\alpha\cdot\eps\beta) = \eta\alpha\cdot\eta\eps\beta  \overset{(*)}{=} \eta\alpha\cdot\eps\eta\alpha = \eta\alpha \neq 0$ since $\eta\alpha$ is an atom in $\eta\Phi$ by Lemma \ref{preservation of atoms}, so that $\eta\alpha\cdot\eps\beta\neq 0$. Thus there exists $\xi\in At\Phi$ such that $\xi\geq\eta\alpha\cdot\eps\beta$.  Now $\eta\xi\geq\eta(\eta\alpha\cdot\eps\beta) = \eta\alpha\cdot\eta\eps\beta \geq \eta\alpha$, that is, $\eta\xi\geq\eta\alpha$. But this implies $\eta\xi = \eta\alpha$ since both $\eta\xi$ and $\eta\alpha$ are atoms in $\eta\Phi$.

Analogously, one obtains $\eps\xi = \eps\beta$, and so $\alpha\equiv'_\eta\star\equiv'_\eps\beta$.
\end{proof}

Assume $\underline{A}$ is atomic and put $\mathcal{E}' := \{\equiv'_\eps: \eps\in E\}$, and let $Sat(\mathcal{E}')$ be the set of all of all saturation operators associated with the equivalences $\equiv'_\eps$. The preceding lemma shows that $(\mathcal{E}',\star)$ is a $\star$-semigroup $\underline{\mathcal{E}}'$, so $\underline{Sat}(\mathcal{E}') := (Sat(\mathcal{E}'),\circ)$ is a commutative idempotent semigroup by Prop. \ref{semigroup of saturation operators}. Let $\underline{P}^d_{red}(At\Phi)$ be the $(\cap,At\Phi,\emptyset)$-reduct of $\underline{P}^d(At\Phi)$.  The operators $\sigma'_\eps$ map $P(At\Phi)$ into itself by definition, and $\sigma'_\eps(At(\Phi)) = At(\Phi)$, so we have a set algebra based on $At\Phi$ at hand.

\begin{theorem} \label{th:DomFreeAtomicHomom}
Let $\underline{A}$ be atomic. Then $(\underline{P}^d_{red}(At\Phi);\underline{Sat}(\mathcal{E}'))$ is an information algebra isomorphic to $SetAlg(\underline{P}^d_{red}(At\Phi);\underline{\mathcal{E}}')$.
\end{theorem}

This is the type of set algebras into which atomistic information algebras may be embedded or to which completely atomistic information algebras are isomorphic. In fact, we have the following representation theorem:

\begin{theorem} \label{th:AtomRepTheo}
Let $\underline{A} = (\underline{\Phi};\underline{E})$ be an atomic  information algebra. Then the pair of maps $ \verb"at": \phi \mapsto At(\phi)$ and  $\verb"j": \epsilon \mapsto \sigma'_\eps$ defines a homomorphism from $\underline{A}$ to $(\underline{P}^d_{red}(At\Phi);\underline{Sat}(\mathcal{E}'))$. If $\underline{A}$ is atomistic, the pair $(\verb"at",\verb"j")$ is an embedding, and if $\underline{A}$ is completely atomistic, it is an isomorphism.
\end{theorem}

\begin{proof}
We verify that
\begin{enumerate}
\item $At(\phi \cdot \psi) = At(\phi) \cap At(\psi)$, $At(1) = At\Phi$, and $At(0) = \emptyset$,
\item $\sigma'_{\eps\circ\eta} = \sigma'_\eps \circ \sigma'_\eta$, for all $\eps,\eta\in E$
\item $\sigma'_\eps At(\phi) = At(\eps\phi)$.
\end{enumerate}

For 1.: $At(1) = At\Phi$ and $At(0) = \emptyset$ are obvious. Since the algebra is atomic, $At(\phi)\neq\emptyset$ if $\phi \neq 0$. Assume $\phi \cdot \psi \neq 0$ and let $\alpha \in At(\phi\cdot \psi)$, thus $\phi,\psi \leq \phi \cdot \psi \leq \alpha$ and $\alpha \in At(\phi) \cap At(\psi)$. Conversely, let $\alpha  \in At(\phi) \cap At(\psi)$. Then $\phi,\psi \leq \alpha$, hence $\phi \cdot \psi \leq \alpha$ and therefore $\alpha \in At(\phi\cdot \psi)$. This shows that $At(\phi \cdot \psi) = At(\phi) \cap At(\psi)$. If $\phi \cdot \psi = 0$, then  $At(\phi \cdot \psi) = \emptyset$ and $At(\phi) \cap At(\psi) = \emptyset$.

For 2.: Theorem \ref{kernels commute} and Lemma \ref{commuting equivalences}

For 3. Assume first that $\alpha \in \sigma'_\eps At(\phi)$. So there exists $\beta \in At(\phi)$ such that $\epsilon\alpha = \epsilon\beta$. $\beta \in At(\phi)$ implies $\phi\leq \beta$, so $\eps\phi \leq \eps\beta = \eps\alpha \leq \alpha$ and thus $\alpha\in At(\eps\phi)$.

Conversely, let $\alpha\in At(\eps\phi)$, thus $\eps\phi \leq \alpha$. Recall that $\phi \leq \epsilon(\alpha) \cdot \phi$. We claim that $\epsilon\alpha \cdot \phi \not= 0$. Indeed, otherwise we would have $\epsilon(\alpha \cdot \epsilon\phi) = \epsilon\alpha \cdot \epsilon\phi = \epsilon(\epsilon\alpha \cdot \phi) = \eps0 = 0$, implying $\alpha \cdot \epsilon\phi = 0$ and contradicting $\alpha \in At(\epsilon\phi)$. So there exists $\beta \in At(\epsilon\alpha \cdot \phi)$, and thus $\phi \leq \epsilon\alpha \cdot \phi \leq \beta$. We conclude that $\beta \in At(\phi)$.

Further $\epsilon(\epsilon\alpha \cdot \phi) = \epsilon\alpha \cdot \epsilon\phi \leq \epsilon\beta$, hence $\epsilon\alpha \cdot \epsilon\beta \cdot \epsilon\phi = \epsilon\beta$. This implies $\epsilon\alpha \cdot \epsilon\beta \not=0$. Since $\epsilon\alpha \cdot \epsilon\beta =\epsilon(\alpha \cdot \epsilon\beta)$ we conclude that $\alpha \cdot \epsilon\beta \not= 0$, hence $\epsilon\beta \leq \alpha$ since $\alpha$ is an atom. We infer that $\epsilon\beta \leq \epsilon\alpha$.

Proceed in the same way from $\epsilon\alpha \cdot \epsilon\beta = \epsilon(\epsilon\alpha \cdot \beta)$ in order to obtain obtain $\epsilon\alpha \leq \epsilon\beta$, and so finally $\epsilon\alpha = \epsilon\beta$. But this means that $\alpha\in\sigma'_\eps At(\phi)$ and so $\sigma'_\eps At(\phi) = At(\eps\phi)$ as claimed.

Finally, the map $\verb"j": \epsilon \mapsto \sigma'_\eps$ is bijective by construction. The map $\verb"at" : \phi \mapsto At(\phi)$ is obviously one-to-one whenever $\underline{A}$ is atomistic, and even onto if $\underline{A}$ is completely atomistic,  concluding the proof.

\end{proof}

For completely atomistic  information algebra there is a much stronger result:.

\begin{theorem} \label{th:ComplAtomBoole}
Let $\underline{A} = (\underline{\Phi};\underline{E})$ be a completely atomistic  information algebra. Then $\underline{\Phi}$ is a complete Boolean lattice, and the map $\verb"at":  \phi \mapsto At(\phi)$ preserves arbitrary joins and meets (in the information order) as well as complements.
\end{theorem}

\begin{proof}
Let $X$ be any subset of $\Phi$ and define

\begin{eqnarray}
A_X = \bigcap_{\phi \in X} At(\phi).
\nonumber
\end{eqnarray}

Assume $A_X \not= \emptyset$. Since the algebra is completely atomistic, there exists $\psi \in \Phi$ such that $A_X = At(\psi)$ and $\psi = \bigwedge A_X$. For any $\alpha \in A_X$ and $\phi \in X$ we have $\phi \leq \alpha$, therefore $\phi \leq \bigwedge A_X$ which shows that $\bigwedge A_X$ is  an upper bound of $X$. Let $\chi$ be any other such upper bound. Then $At(\chi) \subseteq At(\phi)$ for all $\phi \in X$, hence $\alpha \in At(\chi)$ implies $\alpha \in A_X$, and therefore $\chi =
\bigwedge At(\chi) \geq \bigwedge A_X$. It follows $\bigwedge A_X$ is the supremum of $X$, that is, $\bigvee X = \bigwedge A_X$. Consequently,

\begin{eqnarray}
At(\bigvee X) = \bigcap_{\phi \in X} At(\phi). \nonumber
\end{eqnarray}

If $A_X = \emptyset$ , then $\bigvee X = 0$ and $At(0) = \emptyset$. So join (in the information order) is preserved.

Consider $\phi \in \Phi$ and define $At^c(\phi) := At\Phi\setminus At(\phi)$. Since $\underline{A}$ is atomistic, $\psi = \bigwedge At^c(\phi)$ exists and belongs to $\Phi$. Moreover, $At(\psi) = At^c(\phi)$. We know that $At(\phi \cdot \psi) = At(\phi) \cap At(\psi)$, $At(1) = At\Phi$, and $At(0) = \emptyset$ (see the proof of \ref{th:AtomRepTheo}, item 1). Hence

\begin{eqnarray*}
&\phi \vee \psi = \bigwedge At(\phi \cdot \psi)  =  \bigwedge(At(\phi) \cap At(\psi)) =& \nonumber \\
& = \bigwedge (At(\phi) \cap At^{c}(\phi)) = \bigwedge \emptyset = 0&
\end{eqnarray*}

Not unexpectedly, it is true that $At(\phi \wedge \psi) = At(\phi) \cup At(\psi)$,  but this is  where we need $\underline{A}$  to be completely atomistic. Indeed, putting $A := At(\phi) \cup At(\psi)$, then certainly $A\subseteq At(\phi \wedge \psi)$, and there exists $\xi$ such that $A = At(\xi)$ and $\xi = \bigwedge A$. It is clear that $\xi\leq\phi$ and  $\xi\leq\psi$ since $\underline{A}$ is atomistic. If now $\chi\leq \phi$ and $\chi\leq\psi$ for some $\chi$, then certainly $At(\chi)\supseteq A$ and - using "atomistic" in the reverse way - we obtain $\xi = \phi\wedge \psi$ and

\begin{eqnarray*}
&\phi \wedge \psi = \bigwedge (At(\phi) \cup At^{c}(\phi)) = \bigwedge At(\Phi) = 1&
\end{eqnarray*}

So $\psi =: \phi^c$ is the complement of $\phi$ and $At(\phi^c) = At^c(\phi)$.

The map $\verb"at"$ thus preserves arbitrary joins and complements and consequently also arbitrary meets, completing the proof.
\end{proof}

As an illustration we consider string algebras (Section \ref{subsec:Expl}): The infinite strings in $\Sigma^{\omega}$ are the atoms of this algebra. If $s$ is a finite string, then the atoms in $At(s)$ are all infinite strings with $s$ as a prefix. These algebras are atomistic as any string $s$ is the infimum of the set of infinite strings with $s$ as a prefix. But it is not completely atomistic, since there are sets $A$ of infinite strings which do not arise as the set of all atoms over some string $s$, namely sets $A$ containing strings with different prefixes. Thus, the algebra of strings is embedded into the set algebra of its atoms $(\underline{P}^d_{red}(\Sigma^{\omega}),\underline{\underline{}\mathcal{E}}')$ (Theorem \ref{th:AtomRepTheo}) by the map $s \mapsto At(s)$ where, for any $n$, the saturation operator $\sigma'_n$ maps a set $S$ of infinite strings into the set of all infinite strings which have a common prefix of length $n$ with some string from $S$ (compare this representation of the string algebra by sets of infinite strings with the representation of the same algebra by truncated up-sets of arbitrary strings in Section \ref{GenRep Thm}).

%%%%%%%%%%%%%%%%%%%%%%%%%%%%%%%%%%%%%%%%%%%%%%%%%%%%%%%%%%%%%%%%%%%%%

\section{Distributive Information Algebras} \label{sec:DistrLattAlg}

%%%%%%%%%%%%%%%%%%%%%%%%%%%%%%%%%%%%%%%%%%%%%%%%%%%%%%%%%%%%%%%%%%%%%%

%%%%%%%%%%%%%%%%%%%%%%%%%%%%%%%%%%%%%%%%%%%%%%%%%%%%%%%%%%%%%%%%%%%%%

%\subsection{The Finite Case: Birkhoff Duality Extended}

%%%%%%%%%%%%%%%%%%%%%%%%%%%%%%%%%%%%%%%%%%%%%%%%%%%%%%%%%%%%%%%%%%%%%%%%%%

In this section we consider   information algebras where $\Phi$ is a distributive lattice. For distributive lattices, there is a well established representation and duality theory, the so-called Priestley duality theory, generalizing Stone duality for Boolean algebras, see \cite{daveypriestley02}. It will be the base for developing a corresponding representation resp. duality theory for  information algebras based on distributive lattices. Moreover, Cignoli studied existential quantifiers on distributive lattices in \cite{cignoli91}. His results are exactly what we need to extend the representation theory of distributive lattices to  information algebras carried by a distributive lattice.

\begin{definition}  \label{def distr CDF}
An information algebra $\underline{A}=(\underline{\Phi};\underline{E})$ is called {\em distributive} iff
\begin{enumerate}
\item[(i)]in $\Phi$ the infimum (relative to the information order) $\phi\wedge\psi$ exists for all $\phi,\psi\in\Phi$, making $(\Phi;\cdot,\wedge,1,0)$ a lattice,
\item[(ii)] $(\Phi;\cdot,\wedge,1,0)$ is distributive and
\item[(iii)] $\eps(\phi\wedge\psi) =\eps\phi\wedge\eps\psi$ for all $\eps\in E$ and $\phi,\psi\in\Phi$.
\end{enumerate}
\end{definition}

We denote by $\mathbb{D}$ the category of all distributive  information together with homomorphisms $(f,g)$ according to Def. \ref{CDF homos} but {\em subject to the additional condition} that $f$ is also meet-preserving. In particular, the restriction $\eps\vert_{\eps\Phi} := \eps_r$ of $\eps$ to $\eps\Phi$ is such a morphism while $\eps$ is not, in general. This makes $\eps\underline{\Phi}$ a sublattice of $\underline{\Phi}$ - a fact that will play a central r\^{o}le.

\subsection{Adapting Priestley theory to $\mathbb{D}$}

We consider first  an arbitrary distributive lattice $\underline{K} = (K;\vee, \wedge, 0,1)$ in its standard order where $0$ is the least and $1$ the greatest element. The reader is referred to \cite{daveypriestley02} for background and for proofs of facts stated below without justification.

\begin{definition} \label{def:PrimeIdeal}

(i) An ideal $I$ in $K$ is a nonempty down-set $I\subseteq K$ which is closed under join (cf. ``Ideal Completions'' in Section \ref{subsec:Expl}). $I$ is called prime, if $I \neq K$ and whenever $x \wedge y \in I$, then either $x\in I$ or $y \in I$.

(ii) A filter $F$ in $K$ is a nonempty up-set $F\subseteq K$ which is closed under meet. $F$ is called prime, if $F\neq K$ and whenever $x \vee y \in F$, then either $x\in F$ or $y \in F$.
\end{definition}

It is easy to check that a subset $I\subseteq K$ is a prime ideal iff $K\setminus I$ is a prime filter, and vice versa.

\begin{definition} \label{def:MaxIdeal}
An ideal $I \subseteq K$ of is called maximal, if $I \not= K$, and whenever $I\subseteq J$ for some ideal $J\subseteq K$, then either $J = I$ or $J = K$.
\end{definition}

It is easy to check that in an arbitrary bounded distributive lattice $\underline{K}$ every maximal ideal is prime (the converse is not true in, in general - cf. Lemma \ref{prime=max}).

The existence of prime ideals (resp. filters) in arbitrary bounded distributive lattices is not trivial and needs some form of a set existence axiom like the Axiom of Choice (AC) or weaker forms thereof. This is not a real issue except for set theorists. So we take (AC) for granted without reservation and use it in the following (weak) version, tailor-made for our purpose, and labelled as (DPI) in \cite{daveypriestley02}:

(DPI) In a bounded distributive lattice $\underline{K}$, let $J$ be an ideal and $G$ a filter such that $J\cap G = \emptyset$. Then there exist a prime ideal $I\supseteq J$ and a prime filter $F\supseteq G$ such that $I\cap F = \emptyset$ (one my take $F = K \setminus I$, of course).

Let $X(K)$ (or $XK$ to improve readability of formulas) denote the set of all prime ideals of an arbitrary bounded distributive lattice $\underline{K} = (K; \vee,\wedge,0,1)$, and put $W_u = \{I \in XK:u \notin I\}$ for $u\in K$. Further, let $\mathcal{K}^- = \{W_u; u \in K\}$. We have $W_u\cup W_v = W_{u\vee v}$ by the ideal property of the members of $XK$ and $W_u\cap W_v = W_{u\wedge v}$ by their primeness. Also, $W_0 = \emptyset$ and $W_1 = XK$. It follows that $\underline{\mathcal{K}}^- = (\mathcal{K}^-;\cup,\cap,\emptyset,K)$ is a sublattice of the power set lattice $\underline{P}(XK)$.

The basic fact underlying \textit{Priestley duality theory} is that $\underline{K}$ is isomorphic with $\underline{\mathcal{K}}^-$, the isomorphism being given by $u\mapsto W_u$ for all $u\in K$.

We have to adjust the development to our use of information order: In  a set algebra the order is the reverse of set inclusion, and combination - that is: join - in the  information algebra should be represented by set intersection. Technically, this means replacing $\underline{K}$ by its order dual $\underline{K}^d$ and adjusting the isomorphism described above in order to have $\emptyset$ as the greatest and $XK$ as the least element. Consider a prime ideal $I\subseteq K$ and $u\in K$, $u\notin I$: In $\underline{K}^d$, $I$ is prime filter not containing $u$ and thus $K\setminus I$ a prime ideal containing $u$. So we put $X_u = \{I\in XK; u\in I\}$ and $\mathcal{K}^+ = \{X_u; u\in K\}$. We have $X_u\cap X_v = X_{u\vee v}$ by the ideal property and $X_u\cup X_v = X_{u\wedge v}$ by primeness, moreover $X_0 = XK$ and $X_1 = \emptyset$. So $\underline{K}^d$ is isomorphic with $\underline{\mathcal{K}}^+ = (\mathcal{K}^+;\cap,\cup,KX,\emptyset)$, a sublattice of the dual powerset lattice $\underline{P}(XK)^d$, the isomorphism being given by $u\mapsto X_u$ for all $u\in K$.

So far, we have a representation of an arbitrary bounded distributive lattice (resp. its order dual) as a lattice of sets with set intersection and union as operations. However, we don't have, at this point, much insight into the nature of the representing sets $X_u$ (resp. $W_u$). Introducing a suitable topology on $XK$ will provide a very satisfactory solution. Focussing on $\underline{\mathcal{K}}^+$, we consider the family of sets
\begin{eqnarray}
\mathfrak{B}_K = \{X_u\cap (XK \setminus X_v); u,v\in K\}
\nonumber
\end{eqnarray}
which is clearly closed under finite set intersection and thus may serve as an (open) base for a topology $\mathfrak{T}_K$ on $XK$, that is, $\mathfrak{T}_K$ is the collection of all set unions of members of $\mathfrak{B}_K$.  $\mathfrak{T}_K$ is compact and Hausdorff. Compactness implies that the \textit{clopen} (simultaneously open and closed) sets of $\mathfrak{T}_K$ are precisely the members of $\mathfrak
{B}_K$. In order to characterize the sets $X_u$ within the space $(XK,\mathfrak{T}_K)$, we use the natural order on $XK$ given by ordinary set inclusion between prime ideals:  The sets $X_u$  for some $u\in K$ are precisely the clopen \textit{up-sets} of the ordered space $\underline{XK} = (XK,\mathfrak{T}_K,\subseteq)$. Assuming (DPI) one shows that for $I,I'\in XK$ satisfying $I\not\subseteq I'$ there exists a clopen subset  $U\subseteq XK$ such that $I\in U$ but $I'\notin U$. This means that the ordered space $\underline{XK}$ is \textit{totally order-disconnected}. A compact totally order-disconnected ordered space is commonly referred to as a \textit{Priestley space}.

Going back to distributive information algebras $(\underline{\Phi};\underline{E})$ we obtain a {\em representation theorem} for their lattice parts:  $\underline{\Phi}$ is isomorphic with the lattice of all clopen up-sets of the Priestley space $\underline{X\Phi}$, the isomorphism being given by $\phi\mapsto X_\phi$ for all $\phi\in\Phi$.

Mapping $\phi\in\Phi$ to the set of prime ideals containing it, instead of the set of prime ideals excluding it makes sense from the information-theoretic point of view: Prime ideals are consistent complete theories or collections of information elements. Indeed, as ideals they are \textit{consistent} in the sense that they contain with any element all elements implied by it and with any two elements also their combination. Moreover, they are \textit{complete} theories in the sense that if they contain $\phi \wedge \psi$, they must contain $\phi$ or $\psi$. So the map $\phi\mapsto X_{\phi}$ assigns to $\phi$ all consistent and complete theories $I$ which are consistent with $\phi$ (that is, contain $\phi$).

\subsection{``Up-side down'' Priestley duality in a nutshell}   \label{``Up-side down'' Priestley duality}

To facilitate the discussion, we introduce the following notation: For any distributive lattice $\underline{\Phi} = (\Phi;\cdot,\wedge,1,0)$, let $\mathbf{X}\underline{\Phi}$ be the ordered topological space $\underline{X\Phi} = (X\Phi,\mathfrak{T}_{\Phi},\subseteq)$. On the other hand, for any compact totally order-disconnected topological space $\underline{Y} := (Y,\mathfrak{T}, \leq)$ let $L(Y)$ (or $LY$ to improve readability of formulas) be the collection of all clopen up-sets of $\underline{Y}$, and $\mathbf{L}\underline{Y}$ the sublattice of the dual power set lattice $\underline{P}(Y)^d$ induced by $LY$. So the representation for the lattice part $\underline{\Phi}$ of a distributive  information algebra $(\underline{\Phi};\underline{E})$ obtained above takes the simple form $\underline{\Phi} \cong \mathbf{L}\mathbf{X}\underline{\Phi}$. This isomorphism is given explicitly as $\kappa_{\Phi}: \phi\in\Phi \longmapsto X_\phi =\{I\in X\Phi: \phi\in I\}$ for all $\phi\in\Phi$.

Consider any abstract Priestley space $\underline{Y}$, and for any $p\in Y$, let $L_p = \{U\in LY: p\in U\}$. It is easy to check that $L_p$ is a prime ideal in $\mathbf{L}\underline{Y}$. Define a map $\lambda_Y: \underline{Y}\longrightarrow \mathbf{X}\mathbf{L}\underline{Y}$ by $\lambda_Y(p) = L_p$. Priestley duality shows that $\lambda_Y$ is in fact an order-homeomorphism between the spaces $\underline{Y}$ and $\mathbf{X}\mathbf{L}\underline{Y}$, so $\underline{Y}\cong \mathbf{X}\mathbf{L}\underline{Y}$ as Priestley spaces.

Summing up, this establishes a bijective correspondence between bounded distributive lattices on one hand and Priestley spaces, on the other - in fact, essentially the object part of a full categorical equivalence.

Turning to morphisms, consider first two bounded distributive lattices $\underline{\Phi}$ and $\underline{\Psi}$ and let $Hom(\underline{\Phi},\underline{\Psi}$) be the set of all $1$-$0$-preserving lattice homomorphisms from $\underline{\Phi}$ to $\underline{\Psi}$.
%Similarly, for two Priestley spaces $Z$ and $Y$, let $Hom(Z,Y)$ be the set of all continuous order-preserving maps from $Y$ to %$Z$.
Similarly, for two Priestley spaces $\underline{Y}$ and $\underline{Z}$, let $Hom(\underline{Y},\underline{Z})$ be the set of all continuous order-preserving maps from $\underline{Y}$ to $\underline{Z}$.

For $f\in Hom(\underline{\Phi},\underline{\Psi})$ define $\mathbf{X}f\in Hom(\mathbf{X}\underline{\Psi},\mathbf{X}\underline{\Phi})$ by $\mathbf{X}f:I\in\mathbf{X}\underline{\Psi}\longmapsto f^{-1}(I)\in \mathbf{X}\underline{\Phi}$.

For $\alpha\in Hom(\underline{Y},\underline{Z})$ define $\mathbf{L}\alpha \in Hom(\mathbf{L}\underline{Z},\mathbf{L}\underline{Y})$ by $\mathbf{L}:
V\in \mathbf{L}\underline{Z} \longmapsto\alpha^{-1}(V)\in \mathbf{L}\underline{Y}$.

The maps $f\longmapsto \mathbf{X}f$ resp. $\alpha\longmapsto\mathbf{L}\alpha$ provide bijections from $Hom(\underline{\Phi},\underline{\Psi})$ to $Hom(\mathbf{X}\underline{\Psi},\mathbf{X}\underline{\Psi}$ resp. from $Hom(\underline{Y},\underline{Z})$ to $Hom(\mathbf{L}\underline{Z},\mathbf{L}\underline{Y})$.

$f\in Hom(\underline{\Phi},\underline{\Psi})$ is one-to-one iff $\mathbf{X}f\in Hom(\mathbf{X}\underline{\Psi},\mathbf{X}\underline{\Phi})$ is onto, and $f\in Hom(\underline{\Phi},\underline{\Psi})$ is onto iff $\mathbf{X}f\in Hom(\mathbf{X}\underline{\Psi},\mathbf{X}\underline{\Phi})$ is an order embedding.

Finally $\mathbf{L}\mathbf{X}f:\mathbf{L}\mathbf{X}\underline{\Phi}\longrightarrow\mathbf{L}\mathbf{X}\underline{\Psi}$ satisfies $\mathbf{L}\mathbf{X}f\circ\kappa_{\Phi} = f\circ\kappa_{\Psi}$. In particular, we have  $\mathbf{L}\mathbf{X}f(X_{\phi}) = X_{f
(\phi)}$ for all $\phi\in\Phi$.
%Similarly, $\mathbf{X}\mathbf{L}\alpha:\mathbf{X}\mathbf{L}Z^d\longrightarrow\mathbf{X}\mathbf{L}Y^d$ satisfies %$\mathbf{X}\mathbf{L}\alpha\circ\lambda_Z = \lambda_Y\circ\alpha$. Again, one has $\mathbf{X}\mathbf{L}\alpha(L_p) = L_{\alpha %p}$ for all $p\in Z$.
Similarly, $\mathbf{X}\mathbf{L}\alpha:\mathbf{X}\mathbf{L}\underline{Y}\longrightarrow\mathbf{X}\mathbf{L}\underline{Z}$ satisfies $\mathbf{X}\mathbf{L}\alpha\circ\lambda_Y = \lambda_Z\circ\alpha$. Again, one has $\mathbf{X}\mathbf{L}\alpha(L_p) = L_{\alpha(p)}$ for all $p\in Y$.

\subsection{Extraction operators}

Consider any distributive  information algebra $\underline{A} = (\underline{\Phi};\underline{E})$. The key observation is that for any $\eps\in E$, the restriction $\eps_r := \eps\vert_{\eps\underline{\Phi}}$ of $\eps$ to $\eps\underline{\Phi}$ is a $1$-$0$-preserving lattice homomorphism while $\eps$ is not, in general. So we may use Priestley duality to model $\eps_r$ in Priestley spaces. Our exposition is based on the results of \cite{vranckenmawet84} and Cignoli \cite{cignoli91}.

Clearly, $\eps_r:\eps\underline{\Phi}\longrightarrow\underline{\Phi}$ is a one-to-one embedding of $\eps\underline{\Phi}$ into $\underline{\Phi}$. So $\mathbf{X}\eps_r =\eps_r^{-1}:\mathbf{X}(\underline{\Phi})\longrightarrow\mathbf{X}(\eps\underline{\Phi})$ is onto, continuous and order-preserving. Thus $\mathbf{L}\mathbf{X}\eps_r =(\eps_r^{-1})^{-1}: \mathbf{L}\mathbf{X}(\eps\underline{\Phi})\longrightarrow \mathbf{L}\mathbf{X}(\underline{\Phi})$ is a lattice embedding. Note that $\eps_r^{-1}$ takes a prime ideal $I\in X\Phi$ to $I\cap\eps\Phi$ which is a prime ideal in $\eps\Phi$. Also, $(\eps_r^{-1})^{-1}$ takes a clopen up-set $U\in
 L(X(\eps\Phi))$ to $\{I\in X\Phi; I\cap\eps\Phi\in U\}$, this latter being a clopen up-set since $\eps_r^{-1}$ is continuous and order-preserving. Since $\mathbf{L}\mathbf{X}\eps_r\circ \kappa_{\eps\Phi} = \kappa_{\Phi}\circ \eps_r$, we see that $\mathbf{L}\mathbf{X}\eps_r$ takes $X'_{\eps\phi} := \{I'\in X(\eps\Phi):\eps\phi\in I'\}$ to $X_{\eps\phi} = \{I\in X\Phi: \eps\phi\in I\}$ for all $\phi\in\Phi$.

Consider the kernel of $\mathbf{X}\eps_r$, that is $ker\,\mathbf{X} \eps_r = \{ (I,I')\in X\Phi\times X\Phi: I\cap \eps\Phi = I'\cap \eps\Phi\}$, shortly denoted  by $\cong_{\eps}$. We are interested in the saturation operator $\varsigma_{\eps}$ associated  with $\cong_{\eps}$ which for any subset $U\subseteq X\Phi$ returns $\{I'\in X\Phi: I'\cong_{\eps}I \textit{ for some } I\in U\}$, and particularly in the restriction of $\varsigma_{\eps}$ to $L(X\Phi)$.  Sets $U$ satisfying $U = \varsigma_\eps(U)$ will be called \textit{ $\eps$-saturated}. We write $L^\eps(X\Phi)$ for the family of all $\eps$-saturated sets in $L(X\Phi)$.

The key fact we need is the following lemma which is part of Thm. 2.2 in (Cignoli, 1991).

\begin{lemma}[Cignoli] \label{Cignolis lemma}
For any $I\in X\Phi$ containing $\eps\phi \in \Phi$, there exists $I'\in X\Phi$ such that $I'\cong_{\eps}I$ and $\phi\in I'$.
\end{lemma}

Recall that the members of $L(X\Phi)$ are exactly the sets $X_\phi$ for $\phi\in\Phi$. We  have

\begin{corollary}  \label{saturation lemma}
For all $\phi\in\Phi$, $\varsigma_{\eps}X_{\phi} = X_{\eps\phi}$.
\end{corollary}

\begin{proof}
We first show that $X_{\eps\phi}$ is $\varsigma_\eps$-saturated for all $\phi\in\Phi$. Indeed, let $I\in X_{\eps\phi}$ and $I'\cong_{\eps}I$. Now $\eps\phi\in I\cap\eps\Phi = I'\cap\eps\Phi$ and so $\eps\phi\in I'$, that is, $I'\in X_{\eps\phi}$ and thus $\varsigma_\eps(X_{\eps\phi})= X_{\eps\phi}$.

Moreover, $X_{\phi}\subseteq X_{\eps\phi}$ since $\eps\phi\leq\phi$, so $\varsigma_{\eps}(X_{\phi}) \subseteq \varsigma_{\eps}(X_{\eps\phi}) = X_{\eps\phi}$. Let $I\in X_{\eps\phi}$. By Lemma \ref{Cignolis lemma} we find $I'\in X_{\phi}$ such that $I'\cap\eps\Phi = I\cap\eps\Phi$, implying $I\in \varsigma_{\eps}(X_{\phi})$ which shows that $X_{\eps\phi} \subseteq \varsigma_{\eps}(X_{\phi})$. This implies $\varsigma_{\eps}(X_{\eps\phi}) \subseteq \varsigma_{\eps}(X_{\phi})$ so finally $\varsigma_{\eps}(X_{\phi}) = \varsigma_{\eps}(X_{\eps\phi}) = X_{\eps\phi}$.
\end{proof}

The following proposition collects the main properties of the saturation operator $\varsigma_{\eps}$:

\begin{proposition} \label{properties of saturation operators}
Let $\cong_{\eps}$, $\varsigma_{\eps}$ and $L^\eps(X\Phi)$ be given as described above. Then:

(i) $\varsigma_{\eps}$ maps $L(X\Phi)$ into $L(X\Phi)$.

(ii) The members of $L^\eps(X\Phi)$ are exactly the sets $X_{\eps\phi}$ for $\phi\in\Phi$.

(iii) If $I,I'\in X\Phi$ and $I\not\cong_{\eps}I'$, there is $U\in L^\eps(X\Phi)$ containing  exactly one of $I,I'$.

(iv) $\varsigma_\eps$ is an extraction operator on $\mathbf{L}\mathbf{X}\underline{\Phi}$.

(v) $L^\eps(X\Phi)$ endowed with the operations inherited from $\mathbf{L}\mathbf{X}\underline{\Phi}$ is a sublattice $\mathbf{L}^\eps\mathbf{X}\underline{\Phi}$ of $\mathbf{L}\mathbf{X}\underline{\Phi}$ , and $\mathbf{L}^\eps\mathbf{X}\underline{\Phi}\cong \eps\underline{\Phi}$.

(vi) For any $I,I'\in X\Phi$, we have $I\cap\eps\Phi\subseteq I'\cap\eps\Phi$ iff for all $U\in L^\eps(X\Phi)$, $I\in U$ implies $I'\in U$.
\end{proposition}

\begin{proof}
Ad (i): Follows directly from Lemma \ref{saturation lemma}. Ad (ii):  $\varsigma_{\eps}X_{\phi} = X_{\phi}$ iff $X_{\eps\phi} = X_{\phi}$ iff $\eps\phi = \phi$. Ad (iii): Let $I\not\cong_{\eps}I'$, thus $I\cap\eps\Phi\neq I'\cap \eps\Phi$. Assume w.l.o.g. that $I\cap\eps\Phi\nsubseteq I'\cap \eps\Phi$. So there exists $\eps\phi \in I\cap\eps\Phi$ such that $\eps\phi \notin I\cap\eps\Phi$. This means $I\in X_{\eps\phi} \in L^\eps(X\Phi)$ but $I\notin X_{\eps\phi}$. Ad (iv): Lemma \ref{saturation operators}. Ad (v): Reformulates the description of the map $\mathbf{L}\mathbf{X}\eps_r$ given above in terms of $\varsigma_{\eps}$-saturated sets. Ad (vi): Using (ii), the assertion becomes $I\cap\eps\Phi\subseteq I'\cap\eps\Phi$ iff for all $X_{\eps\phi}$, $I\in X_{\eps\phi}$ implies $I'\in X_{\eps\phi}$ iff for all $\eps\phi$, $\eps\phi\in I$  implies $\eps\phi\in I'$ which is the same as  $I\cap\eps\Phi\subseteq I'\cap\eps\Phi$.
\end{proof}

Looking at the other end of the sought duality between algebras and spaces, the obvious question is now how to characterize equivalence relations $\Theta$ on a Priestley space $\underline{Y}$ such that the associated saturation operators $\sigma_{\Theta}$ induce quantifiers on the lattice $\mathbf{L}\underline{Y}$. Obviously, this requires that $\sigma_{\Theta}$ maps $LY$ into $LY$ (corresponding to Prop. \ref{properties of saturation operators} (i)). For any $p\in Y$, let $[p]_{\Theta}$ be the $\Theta$-class of $p$, and put $Y/\Theta =:\{[p]_{\Theta}: p\in Y\}$. Write $L^{\Theta}(Y)$ for the collection of all $\sigma_{\Theta}$-saturated clopen up-sets of $\underline{Y}$, and $\mathbf{L}^{\Theta}(\underline{Y})$ for the corresponding sublattice of $\mathbf{L}\underline{Y}$.

 We want to equip $Y/\Theta$ with an order and a topology making it a Priestley space such that the canonical projection $\pi_{\Theta}:\underline{Y}\longrightarrow Y/\Theta$ is continuous and order-preserving. Imitating Prop. \ref{properties of saturation operators} (vi), tentatively define $[p]_{\Theta}\leq_{\Theta}[q]_{\Theta}$ iff for all $U\in L^{\Theta}(Y)$, $p\in U$ implies $q\in U$. It is obvious that $\leq_{\Theta}$ is reflexive and transitive.

 The key fact we need here is contained in the following lemma, which is part of Lemma 1.6 in \cite{vranckenmawet84}).

 \begin{lemma}[Vrancken-Mawet]
  A Priestley structure on $Y/{\Theta}$ making $\pi_{\Theta}:\underline{Y}\longrightarrow Y/\Theta$ continuous and order-preserving exists exactly if  $\leq_{\Theta}$ is antisymmetric and thus an order. If this is the case, the sought topology on $Y/{\Theta}$ is uniquely determined as the quotient topology relative to $\underline{Y}$ and $\pi_{\Theta}$.
 \end{lemma}

 Consider $p,q\in Y$ such that $(p,q)\not\in\Theta$, that is, $[p]_{\Theta} \neq [q]_{\Theta}$. Now $\Theta$ is antisymmetric iff this implies $[p]\not\leq_{\Theta}[q]$ or $[q]\not\leq_{\Theta}[p]$. According to the definition of $\leq_{\Theta}$ this means that there exists  $U\in L^{\Theta}(Y)$ containing exactly one of $p$ and $q$ (note that this corresponds to Prop. \ref{properties of saturation operators} (iii)).

 \begin{definition}  \label{separating equivalences}
 An equivalence $\Theta$ on a Priestley space $\underline{Y}$ is  {\em separating} iff  (i) $\sigma_{\Theta}$ maps $L(Y)$ into $L(Y)$ and (ii) for any $p,q\in Y$ with $(p,q)\not\in\Theta$ there exists $U\in L^{\Theta}(Y)$ containing exactly one of $p$ and $q$ .
 \end{definition}

 Note that the equivalences $\cong_\eps$ considered in Prop. \ref{properties of saturation operators} are kernels and separating. It remains to see that all separating equivalences on a Priestley space arise as kernels of this type.

Assume that a Priestley $\underline{Y}$ space carries a separating equivalence relation $\Theta$, so that $Y/\Theta$ ordered by $\leq_{\Theta}$ and equipped by the quotient topology relative to $\pi_{\Theta}$ is a Priestley space $\underline{Y}/\Theta$, and $\pi_{\Theta}$ is order-preserving and continuous. Putting Priestley duality to work, we see that $\mathbf{L}\pi_{\Theta} = \pi_{\Theta}^{-1}$ takes $U\in L(Y/\Theta)$ to $\pi_{\Theta}^{-1}(U)\in L^{\Theta}(Y)\subseteq LY$, is one-to-one and thus provides a lattice isomorphism between $\mathbf{L}(\underline{Y}/\Theta)$ and $\mathbf{L}^{\Theta}(\underline{Y})$. Note that $\pi_{\Theta}^{-1}\pi_{\Theta}$ is an extraction operator on $\mathbf{L}\underline{Y}$ with image $\mathbf{L}^{\Theta}(Y)$.

Now $\mathbf{X}\mathbf{L}\pi_{\Theta} = (\pi_{\Theta}^{-1})^{-1}$ takes a prime ideal $I\in XL(Y)$ to $I\cap L^{\Theta}(Y)\in XL(Y/\Theta)$ and is onto $XL(Y/\Theta)$. Moreover, $ker\, \mathbf{X}\mathbf{L}\pi_{\Theta} = \{(I,I')\in XL(Y)\times XL(Y): I\cap L^{\Theta}(Y) = I'\cap L^{\Theta}(Y)\}$. But the prime ideals in $XL(Y)$ are exactly the sets $L_p = \{U\in L(Y): p\in U\}$. So we obtain

\begin{lemma}  \label{extraction and equivalence}
$(p,q)\in\Theta$ iff $(L_p,L_q)\in ker\, \mathbf{X}\mathbf{L}\pi_{\Theta}$.
\end{lemma}

\begin{proof}
We have $(L_p,L_q)\in ker\, \mathbf{X}\mathbf{L}\pi_{\Theta}$ iff $L_p\cap L^{\Theta} =  L_q\cap L^{\Theta}$ iff $(p,q)\in\Theta$ since $\Theta$ is separating.
\end{proof}

Since $\underline{Y}\cong\mathbf{X}\mathbf{L}(\underline{Y})$ as Priestley spaces, we conclude that $\Theta$ indeed corresponds to the kernel of $\mathbf{X}\mathbf{L}\pi_{\Theta}$ under this isomorphism.

\begin{corollary}  \label{quantifiers vs equivalences}
There is a bijective correspondence between meet-preserving extraction operators on a bounded distributive lattice and separating equivalence relations on its Priestley space.
\end{corollary}

It should be noted (see \cite{cignoli91}) that condition (ii) in Def. \ref{separating equivalences} could be replaced by requiring the equivalence classes of $\Theta$ to be topologically closed in $\underline{Y}$. While undoubtedly more elegant, this approach does not immediately reveal how the condition actually is put to work.

\subsection{$Q$-Priestley spaces}    \label{$Q$-Priestley spaces}

In order to  extend Priestley duality theory of distributive lattices to lattices with a quantifier, \cite{cignoli91} introduced the concept of a $Q$-\textit{space}. We will extend this concept further in order to obtain, in the end, a full duality theory for distributive  information algebras. In view of Cor. \ref{quantifiers vs equivalences} the dual object of a distributive  information algebra $\underline{A} = (\underline{\Phi};\underline{E})$ should obviously be a Priestley space equipped with a collection of commuting separating equivalence relations.

\begin{lemma}  \label{varsigma is a homomorphism}
$\varsigma_\eps\circ\varsigma_\eta = \varsigma_{\eps\circ\eta}$ for any $\eps,\eta\in E$.
\end{lemma}

\begin{proof}
By Corollary \ref{Cignolis lemma} we have $\varsigma_\eps X_\phi = X_{\eps\phi}$. So $(\varsigma_\eps \circ\varsigma_\eta) X_\phi = \varsigma_\eps(\varsigma_\eta X_\phi) = \varsigma_\eps X_{\eta\phi} = X_{\eps\circ\eta(\phi)} = \varsigma_{\eps\circ\eta}X_\phi$.
\end{proof}

\begin{lemma}   \label{varsigma is one-to-one}
If $\eps\neq\eta$ in $E$, then $\varsigma_\eps\neq\varsigma_\eta$.
\end{lemma}

\begin{proof}
If $\eps\neq\eta$, there is $\phi\in\Phi$ such that $\eps\phi\neq\eta\phi$. By (DPI), we find $I\in X\Phi$ such that w.l.o.g. $\eps\phi\in I$ but $\eta\phi\notin I$. This means that $I\in X_{\eps\phi} = \varsigma_\eps X_\phi$ but $I\notin X_{\eta\phi} = \varsigma_\eta X_\phi$ using Lemma  \ref{varsigma is a homomorphism}, so $\varsigma_\eps  X_\phi \neq\varsigma_\eta X_\phi $.
\end{proof}

{\bf Note}: $X_\phi$ used in the proof above is a member of $L(X\Phi)$. This means that even the restrictions $\varsigma_{\eps}\vert_{L(X\Phi)}$ and $\varsigma_{\eta}\vert_{L(X\Phi)}$ differ whenever $\eps\neq\eta$.

We extend the $\mathbf{X}$-$\mathbf{L}$-machinery in order to include extraction and define $X(E)$ (or just $XE$ for short) for $\underline{A} = (\underline{\Phi};\underline{E})$ by $XE = \{\cong_\eps:\eps\in E\}$ and $\mathbf{X}\underline{E} := (XE,\star)$. Also, let $Sat_{L(X\Phi)}(XE) := \{\varsigma_{\eps}\vert_{L(X\Phi)}: \eps\in E\}$ and $\underline{Sat}_{L(X\Phi)}(XE) := (Sat_{L(X\Phi)}(XE),\circ)$.

\begin{theorem}  \label{ E is EX is SatEX}
$\underline{E}$, $\underline{Sat}_{L(X\Phi)}(XE)$ and $\mathbf{X}\underline{E}$ are isomorphic as semigroups.
\end{theorem}

\begin{proof}
Lemma \ref{varsigma is a homomorphism} and Lemma \ref{varsigma is one-to-one} for the first isomorphism, and
Prop. \ref{semigroup of saturation operators} for the second.
\end{proof}

\begin{definition} \textbf{$Q$-Priestley Spaces:} \label{Q-Priestley Spaces}
A $Q$-Priestley space is a pair $(\underline{Y},\underline{\mathcal{T}})$ consisting of a Priestley space $\underline{Y}$ and a $\star$-semigroup $\underline{\mathcal{T}}$ in $Eq(Y)$ consisting of separating equivalence relations.
\end{definition}

Given a Q-Priestley space $(\underline{Y},\underline{\mathcal{T}})$, we extend notation again and write $L\mathcal{T} := \{\sigma_{\Theta}\vert_{LY}: \Theta\in\mathcal{T}\} = Sat_{LY}(\mathcal{T})$, and $\mathbf{L}\underline{\mathcal{T}} := (L\mathcal{T},\circ)$.

The task at hand is to find the appropriate morphisms between Q-Priestley spaces with the objective of obtaining a full duality between distributive  information algebras with their algebra homomorphisms and $Q$-Priestley spaces with the morphisms sought after.

So let $\underline{A} = (\underline{\Phi};\underline{E})$ and $\underline{B} = (\underline{\Psi};\underline{D})$ be two distributive  information algebras. Assume $(f,g):\underline{A}\longrightarrow\underline{B}$ is an information algebra homomorphism, which means that $f$ is lattice homomorphism and $g$ a semigroup homomorphism subject to the compatibility condition  Def. 2.6.(4). Going to spaces, we have $\mathbf{X}f = f^{-1}:\mathbf{X}\Psi\longrightarrow\mathbf{X}\Phi$ for the lattice part. For the extraction part, the canonical map naturally associated with $g$ is $\mathbf{X}g:\mathbf{X}\underline{E}\longrightarrow \mathbf{X}\underline{D}$ given by $\mathbf{X}g(\cong_{\eps}) := \ \cong_{g(\eps)}$.

\begin{lemma}   \label{Xg is semigroup homo}
$\mathbf{X}g$ is a semigroup homomorphism between the $\star$-semigroups $\mathbf{X}\underline{E}$ and $\mathbf{X}\underline{D}$.
\end{lemma}

\begin{proof}
We have $\mathbf{X}g(\cong_\eps\star\cong_{\eps'}) = \mathbf{X}g(\cong_{\eps\circ\eps'}) =\ \cong_{g(\eps\circ\eps')} =\  \cong_{g\eps\circ g\eps'} =\  \cong_{g\eps}\star\cong_{g\eps'} = \mathbf{X}g(\cong_{\eps})\star\mathbf{X}g(\cong_{\eps'})$.
\end{proof}

Proceeding in the obvious way, define $\mathbf{L}\mathbf{X}g: L(XE) \longrightarrow L(XD)$ by $\mathbf{L}\mathbf{X}g(\varsigma_\eps) := \varsigma_{g(\eps)}$.

\begin{lemma}   \label{LXg is semigroup homo}
$\mathbf{L}\mathbf{X}g$ is semigroup homomorphism between the semigroups $\mathbf{L}\mathbf{X}\underline{E} := (L(XE),\circ)$ and $\mathbf{L}\mathbf{X}\underline{D} := (L(XD),\circ)$.
\end{lemma}

\begin{proof}
Theorem \ref{ E is EX is SatEX}
\end{proof}

Next, we will show that the pair of maps $(\mathbf{L}\mathbf{X}f,\mathbf{L}\mathbf{X}g)$ is an information algebra homomorphism from the algebra $\mathbf{L}\mathbf{X}\underline{A} = (\mathbf{L}\mathbf{X}\underline{\Phi},\mathbf{L}\mathbf{X}\underline{E})$ to the algebra $\mathbf{L}\mathbf{X}\underline{B} = (\mathbf{L}\mathbf{X}\underline{\Psi},\mathbf{L}\mathbf{X}
\underline{D})$, which means that $(\mathbf{L}\mathbf{X}f,\mathbf{L}\mathbf{X}g)$ satisfies the compatibility condition Def. 2.6.(4). Let $U\in L(X\Phi)$. So $U = X_\phi$ for some uniquely determined $\phi\in\Phi$. Now
\begin{eqnarray}   \label{alghomo uno}
\mathbf{L}\mathbf{X}f(\varsigma_\eps X_\phi) = \mathbf{L}\mathbf{X}f(X_{\eps(\phi)}) = X_{f(\eps(\phi))}.
\end{eqnarray}
On the other hand,
\begin{eqnarray}   \label{alghomo dos}
\mathbf{L}\mathbf{X}g(\varsigma_\eps)(\mathbf{L}\mathbf{X}f(X_\phi)) = \varsigma_{h(\eps)}(X_{f(\phi)}) = X_{g(\eps)(f(\phi))}.
\end{eqnarray}
Hence $\mathbf{L}\mathbf{X}f(\varsigma_\eps X_\phi) = \mathbf{L}\mathbf{X}g(\varsigma_\eps)(\mathbf{L}\mathbf{X}f(X_\phi))$ iff
$X_{f(\eps(\phi))}) = X_{g(\eps)(f(\phi))}$.

\bigskip

\begin{proposition}   \label{f,g compatible iff LXf,LXg compatible}
$(\mathbf{L}\mathbf{X}f,\mathbf{L}\mathbf{X}g)$ satisfies Def. 2.6.(4) iff $(f,g)$ so does.
\end{proposition}

Turning to spaces, a morphism $(\alpha,\omega)$ from a $Q$-Priestley space $(\underline{Y},\underline{\mathcal{T}})$ to a $Q$-Priestley space $(\underline{Z},\underline{\mathcal{G}})$ should obviously be a pair $(\alpha,\omega)$ consisting of a continuous order-preserving map $\alpha: \underline{Y}\longrightarrow \underline{Z}$ and $\star$-homomorphism $\omega:\underline{\mathcal{G}}\longrightarrow\underline{\mathcal{T}}$. Define  $\mathbf{L}\omega:\mathbf{L}\underline{\mathcal{G}}\longrightarrow \mathbf{L}\underline{\mathcal{T}}$ by $\mathbf{L}\omega(\sigma_\Gamma) := \sigma_{\omega\Gamma}$. $\mathbf{L}\omega$ is a $\circ$-homomorphism by Prop. \ref{semigroup of saturation operators}.

Obviously, we want $(\mathbf{L}\alpha,\mathbf{L}\mathbf{\omega})$ to be an algebra homomomorphism from $\mathbf{L}\underline{Z}$ to $\mathbf{L}\underline{Y}$. This is the case exactly iff $(\mathbf{L}\alpha,\mathbf{L}\omega)$ satisfies Def. 2.6.(4), explicitly,
\begin{eqnarray}   \label{def Q-homos}
\mathbf{L}\alpha(\sigma_{\Gamma}(V)) = \sigma_{\omega\Gamma}(\mathbf{L}\alpha(V))
\end{eqnarray}
for all $V\in LZ$ and $\Gamma\in\mathcal{G}$.

Finally, put $\simeq_\Gamma := ker\ \sigma_\Gamma^{-1}$ for $\Gamma\in\mathcal{G}$ resp. $\simeq_\Theta := ker\ \sigma_\Theta^{-1}$ for $\Theta\in\mathcal{T}$ and let $\mathbf{X}\mathbf{L}\mathcal{G} := \{\simeq_\Gamma : \sigma_\Gamma\in\mathbf{L}\mathcal{G}\}$ resp. $\mathbf{X}\mathbf{L}\mathcal{T} := \{\simeq_\Theta : \sigma_\Theta\in\mathbf{L}\mathcal{T}\}$. Define a map $\mathbf{X}\mathbf{L}\omega: \mathbf{X}\mathbf{L}\underline{\mathcal{G}}\longrightarrow\mathbf{X}\mathbf{L}\underline{\mathcal{T}}$ by $\mathbf{X}\mathbf{L}\omega(\simeq_\Gamma) := \simeq_{\omega\Gamma}$ for all $\sigma_\Gamma\in\mathbf{L}\underline{\mathcal{G}}$. $\mathbf{X}\mathbf{L}\omega$ is a $\star$-homomorphism by Lemma \ref{extraction and equivalence}.

Note that (\ref{alghomo uno}) and (\ref{alghomo dos}) together just say that $\mathbf{X}f =: \alpha$ and $\mathbf{X}g =: \omega$ satisfy (\ref{def Q-homos}). So (\ref{def Q-homos}) is indeed the correct Q-Priestley space analogue of the algebra compatibility condition Def. 2.6.(4) and we formally define

\begin{definition}   \label{def Q-morphism}
A $Q$-morphism $(\alpha,\omega)$ from a $Q$-Priestley space $(\underline{Y},\underline{\mathcal{T}})$ to a $Q$-Priestley space $(\underline{Z},\underline{\mathcal{G}})$ is a pair $(\alpha,\omega)$ consisting of a continuous order-preserving map $\alpha:
\underline{Y}\longrightarrow \underline{Z}$ and $\star$-homomorphism $\omega:\underline{\mathcal{G}}\longrightarrow \underline{\mathcal{T}}$ satisfying  $\mathbf{L}\alpha(\sigma_{\Gamma}(V)) = \sigma_{\omega\Gamma}(\mathbf{L}\alpha(V))$ for all $V\in LZ$ and $\Gamma\in\mathcal{G}$.
\end{definition}

\subsection{Representation and Duality}   \label{Full duality}

Remember that $\mathbb{D}$ stands for the category of all distributive  information algebras with CDF homomorphisms, and write $\mathbb{Q}$ for the category of all $Q$-Priestley spaces with $Q$-morphisms. A full duality between $\mathbb{D}$ and $\mathbb{Q}$ will be established by two commutative diagrams generalizing these given in \cite{daveypriestley02} for the categories of distributive bounded lattices with $0$-$1$-preserving lattice homomorphisms and Priestley spaces with continuous order-preserving maps.

We start with the algebra point of view where the definition of an isomorphism is the natural one (see section  \ref{Homomorphisms and Subalgebras}).

For any  information algebra $\underline{A} = (\underline{\Phi},\underline{E})$, we have - by up-side down Priestley duality -
a natural $1$-$0$-preserving lattice isomorphism $\kappa_{\Phi}:  \underline{\Phi}\longrightarrow \mathbf{L}\mathbf{X}\underline{\Phi}$, given by $\kappa_{\Phi}(\phi) = X_{\phi}$ for all $\phi\in\Phi$ (see section \ref{``Up-side down'' Priestley duality}).
For the extraction part, define a map $\kappa_E: \underline{E}\longrightarrow \mathbf{L}\mathbf{X}\underline{E}$ by $\kappa_E(\eps) := \varsigma_{\eps}$, which is a semigroup isomorphism by Thm. \ref{ E is EX is SatEX}. It remains to show that $(\kappa_{\Phi},\kappa_E)$ satisfies Def. 2.6.(4): We have $\kappa_{\Phi}(\eps(\phi)) = X_{\eps(\phi)} = \varsigma_{\eps}(X_{\phi}) = \kappa_E(\eps)(\kappa_{\Phi}(\phi))$ by Corollary \ref{saturation lemma}, so $(\kappa_{\Phi},\kappa_E)$ is indeed an isomorphism of CDF information algebras by Corollary \ref{alg isos 2}. The same is true for $(\kappa_{\Psi},\kappa_D)$. $(\mathbf{L}\mathbf{X}f,\mathbf{L}\mathbf{X}g)$ is an algebra homomorphism by Prop. \ref{f,g compatible iff LXf,LXg compatible}, so the following diagram is commutative, providing the algebra half of sought duality:

\bigskip

\[
  \begin{CD}
      (\underline{\Phi};\underline{E})                                             @>f>g>                                            (\underline{\Psi};\underline{D})\\
      @V\kappa_{\Phi}V\kappa_EV                                                                         @V\kappa_{\Psi}V\kappa_DV\\
      (\mathbf{L}\mathbf{X}\underline{\Phi};\mathbf{L}\mathbf{X}\underline{E})     @>\mathbf{L}\mathbf{X}f>\mathbf{L}\mathbf{X}g>                (\mathbf{L}\mathbf{X}\underline{\Psi};\mathbf{L}\mathbf{X}\underline{D})
   \end{CD}
\]

\bigskip

Ignoring the horizontal arrows in the preceding diagram, we obtain a general representation theorem:

\begin{theorem}[Representation Theorem]  \label{Gen Repr Thm Distr InfAlgs}

Any distributive CFD information algebra $(\underline{\Phi};\underline{E})$ is isomorphic with the set algebra
$(\mathbf{L}\mathbf{X}\underline{\Phi};\mathbf{L}\mathbf{X}\underline{E})$.

\end{theorem}

For the space analogue, we need a workable description of $Q$-isomorphisms, taking over the r\^{o}le of Corollary \ref{alg isos 2}. Such is provided by an appropriate extension of Corollary 2.9 in (Cignoli, 1991):

\begin{lemma}[Cignoli]    \label{Q-isos}
A $Q$-morphism $(\alpha,\omega):(\underline{Y},\underline{\mathcal{T}})\longrightarrow (\underline{Z},\underline{\mathcal{G}})$ is a $Q$-isomorphism iff $\alpha$ is an order-homeomorphism, $\omega$ is a semigroup isomorphism, and for all $\Theta\in\mathcal{T}$ and all $p,q\in Y$ we have $(p,q)\in\Theta$ iff $(\alpha(p),\alpha(q))\in\omega\Gamma$.
\end{lemma}

For any $Q$-Priestley space $(\underline{Y},\underline{\mathcal{T}})$, we have - by up-side down Priestley duality -
a natural order homeomorphism $\lambda_Y: \underline{Y}\longrightarrow \mathbf{X}\mathbf{L}\underline{Y}$ given by $\lambda_Y(p) = L_p$ for all $p\in Y$ (see section \ref{``Up-side down'' Priestley duality}). For the extraction part, define a map $\lambda_{\mathcal{T}}: \underline{\mathcal{T}} \longrightarrow \mathbf{X}\mathbf{L}\underline{\mathcal{T}}$ by $\lambda_{\mathcal{T}}(\Theta) :=\, \simeq_{\Theta}$, which is a semigroup isomorphism (cf. the proof of Thm. \ref{ E is EX is SatEX}). Consider $p,q\in Y$. We have $\lambda_Y(p)\simeq\lambda_Y(q)$ iff $L_p\simeq_{\Theta}L_q$ iff $\sigma_{\Theta}^{-1}(L_p) = \sigma_{\Theta}^{-1}(L_q)$.  Now
$\sigma_{\Theta}^{-1}(L_p) = \{U\in LY: \sigma_{\Theta}(U)\in L_p\} = \{U\in LY: p\in \sigma_{\Theta}(U)\}$. So $L_p\simeq_{\Theta}L_q$ iff $\{U\in LY: p\in \sigma_{\Theta}(U)\}=\{U\in LY: q\in \sigma_{\Theta}(U)\}$. This is equivalent with $(p,q)\in\Theta$ since $\Theta$ is separating (cf. Lemma  \ref{extraction and equivalence}). By Cignoli's Lemma above it follows that $(\lambda_Y,\lambda_{\mathcal{T}})$ is a $Q$-isomorphism - and with that, also $(\lambda_Y^{-1},\lambda_{\mathcal{T}}^{-1})$. The same goes for $(\lambda_Z,\lambda_{\mathcal{G}})$ and $(\lambda_Z^{-1},\lambda_{\mathcal{G}}^{-1})$, of course. The diagram below is commutative by construction, so - using these isomorphisms - we see that $(\alpha,\omega)$ satisfies (\ref{def Q-homos}) iff $(\mathbf{X}\mathbf{L}\alpha,\mathbf{X}\mathbf{L}\omega)$ so does, establishing the space half of the sought duality.

\bigskip

\[
  \begin{CD}
      (\underline{Y},\underline{\mathcal{T}})      @>\alpha>\omega>    (\underline{Z},\underline{\mathcal{G}})\\
      @V\lambda_{Y}V\lambda_{\mathcal{T}}V                                                                       @V\lambda_{Z}V\lambda_{\mathcal{G}}V\\
      (\mathbf{X}\mathbf{L}\underline{Y},\mathbf{X}\mathbf{L}\underline{\mathcal{T}})     @>\mathbf{X}\mathbf{L}\alpha>\mathbf{X}\mathbf{L}\omega>                (\mathbf{X}\mathbf{L}\underline{Z},\mathbf{X}\mathbf{L}\underline{\mathcal{G}})
   \end{CD}
\]

\bigskip

So the two commutative diagrams together establish

\begin{theorem}
The functors $\mathbf{X}$ and $\mathbf{L}$ induce a full duality between the categories $\mathbb{D}$ and $\mathbb{Q}$.
\end{theorem}

So distributive  information algebras and Q-Priestley spaces are two sides of the same coin.

\subsection{Boolean Information Algebras}   \label{Boolean CDFs}

Consider a bounded distributive lattice $\underline{\Phi} = (\Phi;\cdot,\wedge,1,0)$. For $\phi\in\Phi$, an element $\psi\in\Phi$ is a {\em complement} of $\phi$ iff $\phi\cdot\psi = 0$ and $\phi\wedge \psi = 1$. Using distributivity, it is not hard to see that complements, whenever they exist in $\Phi$, are uniquely determined. If every $\phi\in\Phi$ has a (unique) complement, the $\Phi$ is called {\em complemented}. A complemented distributive lattice is commonly referred to as a {\em Boolean lattice}.

 This must be distinguished from a {\em Boolean algebra} which is Boolean lattice where the operation $\phi\longmapsto \phi^c$ with $\phi^c$ the complement of $\phi$ is a fundamental operation. Boolean algebras thus are structures of type $(\cdot,\wedge,^c,1,0)$. Using distributivity, it is easy to check that a 1-0-preserving lattice homomorphism between Boolean lattices automatically also preserves complements. There is more:

\begin{lemma}   \label{Boolean lattices yield distributive CDFs}
Let $\underline{A} = (\underline{\Phi};\underline{E})$ be  information algebra with $\underline{\Phi}$ a Boolean lattice. Then $\underline{A}$ is distributive  information algebra in the sense of Def. \ref{def distr CDF}.
\end{lemma}

\begin{proof}
It suffices to show that item (iii) of Def. \ref{def distr CDF} is satisfied. Recall that for $\phi,\psi\in\Phi$ we have $\psi \leq \phi$ iff $\phi \cdot \psi^{c} = 0$ (*) in any Boolean lattice. Consider any $\phi,\psi\in\Phi$ and put $\eta = \phi \wedge \psi$.  Then $\eta \leq \phi,\psi$ implies $\epsilon(\eta) \leq \epsilon(\phi)$ and $\epsilon(\eta) \leq \epsilon(\psi)$. Hence $\epsilon(\eta)$ is a lower bound of $\epsilon(\phi)$ and $\epsilon(\psi)$.

Let $\chi$ be another lower bound of $\epsilon(\phi)$ and $\epsilon(\psi)$. Then by (*) above, $\epsilon(\phi) \cdot \chi^{c} = 0$ and $\epsilon(\psi) \cdot \chi^{c} = 0$. It follows that
\begin{eqnarray}
0 = \epsilon(0) = \epsilon(\epsilon(\phi) \cdot \chi^{c}) = \epsilon(\phi) \cdot \epsilon(\chi^{c})
= \epsilon(\phi \cdot \epsilon(\chi^{c})).
\nonumber
\end{eqnarray}
This implies $\phi \cdot \epsilon(\chi^{c}) = 0$. In the same way we obtain $\psi \cdot \epsilon(\chi^{c}) = 0$. Using distributivity and remembering that combination is join, we get
\begin{eqnarray}
0 &=& (\phi \cdot \epsilon(\chi^{c})) \wedge (\psi \cdot \epsilon(\chi^{c}))
= (\phi \wedge \psi) \cdot \epsilon(\chi^{c}) = \eta \cdot \epsilon(\chi^{c}).
\nonumber
\end{eqnarray}
It follows that
\begin{eqnarray}
0 = \epsilon(0) = \epsilon(\eta \cdot \epsilon(\chi^{c})) = \epsilon(\eta) \cdot \epsilon(\chi^{c})
= \epsilon(\epsilon(\eta) \cdot \chi^{c}),
\nonumber
\end{eqnarray}
hence $\epsilon(\eta) \cdot \chi^{c} = 0$. But this implies $\chi \leq \epsilon(\eta)$ by(*) and $\epsilon(\eta)$ is thus the greatest lower bound of $\epsilon(\phi)$ and $\epsilon(\psi)$, that is, $\epsilon(\phi \wedge \psi) = \epsilon(\phi) \wedge \epsilon(\psi)$ as claimed.
\end{proof}

Accordingly, we define a {\em Boolean}  information to be an information algebra $\underline{A} = (\underline{\Phi};\underline{E})$ where $\underline{\Phi}$ is a Boolean lattice.

\begin{lemma}     \label{prime=max}
In a Boolean lattice $\underline{\Phi}$, prime ideals are maximal.
\end{lemma}

\begin{proof}
Let $I\subseteq\Phi$ be a prime ideal, and $\phi\notin I$. Now $\phi\wedge\phi^c = 1\in I$, so $\phi^c\in I$ by primeness of $I$. Let $I'$ be the ideal generated by $I\cup\{\phi\}$ in $\Phi$. Then $\phi,\phi^c\in I$ and thus $1 = \phi\cdot\phi^c\in I'$,which implies $I'=\Phi$. So $I$ is maximal as claimed.
\end{proof}

\begin{corollary}  \label{Priestley = Stone}
The Priestley space $\mathbf{X}\underline{\Phi}$ of a Boolean  information algebra $\underline{A} = (\underline{\Phi};\underline{E})$ carries the trivial order.
\end{corollary}

Let $\underline{A} = (\underline{\Phi};\underline{E})$ be any Boolean  information algebra. Then $\mathbf{X}\underline{\Phi}$ is just a compact Hausdorff space such that for $I,I'\in X\Phi$ satisfying $I\neq I'$ there exists a clopen subset $U\subseteq X\Phi$ with $I\in U$ but $I'\notin U$. This latter property is called {\em total disconnectedness}, and compact Hausdorff totally disconnected spaces are better known as {\em Stone} spaces. Since there is no order to be preserved, the appropriate morphisms between Stone spaces are just continuous maps. Turning to extraction, an equivalence $\Theta$ on a Stone space $\underline{Y}$ will be called {\em separating} iff $\sigma_{\Theta}$ maps clopen subsets of $Y$ to clopen subsets, and for any $p,q\in Y$ with $(p,q)\not\in\Theta$ there exists a clopen subset $U\subseteq Y$ containing exactly one of $p$ and $q$ (cf. Def. \ref{separating equivalences}). Mimicking Def. \ref{Q-Priestley Spaces}, we say that {\em $Q$-Stone space} is a pair $(\underline{Y},\underline{\mathcal{T}})$ consisting of a Stone space $\underline{Y}$ and a $\star$-semigroup $\underline{\mathcal{T}}$ in $Eq(Y)$ consisting of separating equivalence relations. Finally, let $\mathbb{B}$ the category of Boolean  information algebras with CDF homomorphisms, and $\mathbb{QS}$ that of $Q$-Stone spaces with $Q$-morphisms. It immediately follows that

\begin{theorem}
The functors $\mathbf{X}$ and $\mathbf{L}$ induce a full duality between the categories $\mathbb{B}$ and $\mathbb{QS}$.
\end{theorem}

Remember that 1-0-preserving lattice homomorphisms between Boolean lattices also preserve complements. So we could substitute "Boolean lattice" by "Boolean algebra" in the preceding discussion since introducing complementation as an additional fundamental operation does not interfere with extraction.

\subsection{Finite Distributive Information Algebras}    \label{Finite Distr CDFs}

In the preceding subsection, order was trivial on the Priestley space $\mathbf{X}\underline{\Phi}$ associated with a {\em Boolean}  information algebra. A similar situation arises if we consider a distributive  information algebra $\underline{A} = (\underline{\Phi};\underline{E})$ where $\underline{\Phi}$ is {\em finite}: Here the topology of the Priestley space $\mathbf{X}\underline{\Phi}$ is trivial - more precisely: discrete - , being Hausdorff. In plainer terms, $\mathbf{X}\underline{\Phi}$ is just a finite (partially) ordered set $(H,\leq)$. Turning to extraction, $\underline{E}$ is obviously finite and so $\mathbf{X}{\underline{E}}$ is a finite set of equivalence relations on $H$, closed under $\star$ - hence pairwise commuting by Lemma \ref{star-closed} - and subject to the two conditions of Def. \ref{separating equivalences} characterizing separating equivalences.

The point here is that $H$ may be identified with a subset of $\Phi$, which decreases the set-theoretical complexity of the members of $\mathbf{X}\underline{\Phi}$. Indeed, $\Phi$ being finite, the ideals in $\Phi$ are precisely the {\em principal down-sets} $I_\phi = \downarrow\!\phi =\{\psi: \psi\leq \phi\}$ for $\phi\in\Phi$. Call an element $\phi\in\Phi$ {\em meet-irreducible} iff $\phi = \psi_1\wedge\psi_2$ for some $\psi_1,\psi_2\in\Phi$ implies that $\phi=\psi_1$ or $\phi=\psi_2$ (equivalently, iff $\phi$ has exactly one upper neighbor in the order of $\Phi$).

\begin{lemma}   \label{prime=meet-irreducible}
$I_\phi$ is prime iff $\phi$ is meet-irreducible.
\end{lemma}

\begin{proof}
If $\phi = \psi_1\wedge\psi_2$ and $\psi_1\neq\phi\neq\psi_2$, then $I_\phi$ is clearly not prime. So assume $\phi$ is meet-irreducible and $\psi_1\wedge\psi_2\in I$. Then $\psi_1\wedge\psi_2\in I_\phi$, that is, $\psi_1\wedge\psi_2\leq\phi$. Thus $(\psi_1\wedge\psi_2)\vee\phi = \phi =(\psi_1\vee\phi)\wedge(\psi_2\vee\phi)$, using distributivity, and so $\phi = \psi_1\vee\phi$ or $\phi = \psi_2\vee\phi$. But this means $\psi_1\leq\phi$ or $\psi_2\leq\phi$, that is, $\psi_1\in I_\phi$ or $\psi_2\in I_\phi$.
\end{proof}

Let $\mathcal{M}(\underline{\Phi})$ be the set of all meet-irreducibles of $\underline{\Phi}$. Obviously, $\phi\in I_\mu$ iff $\phi\leq\mu$. So $X_\phi = \{I\in X\Phi:\phi\in I\}$ may be identified with $\{\mu\in \mathcal{M}(\underline{\Phi}): \phi\leq \mu\} = \uparrow\!\phi \cap \mathcal{M}(\underline{\Phi})$. So the ordered set $(H,\leq)$ at hand may be concretized as $\mathcal{U}(\mathcal{M}(\underline{\Phi}),\subseteq)$ the final result is

\begin{proposition}    \label{duals of findistrCDFs}
The map $\phi\in\Phi\longmapsto\uparrow\!\phi \cap \mathcal{M}(\underline{\Phi})$ provides a lattice isomorphism between $\underline{\Phi}$ and the lattice of all up-sets in $\mathcal{M}(\underline{\Phi})$, a sublattice of the dual power set lattice $\underline{P}(X\Phi)^d$.
\end{proposition}

For a detailed account, the reader is referred to \cite{daveypriestley02}.

 Focussing on the object part of the duality between  distributive information algebras and  their representing structures, we are left with pairs $(\leq,\mathcal{T})$ where $\leq$ is an order on a finite set $H$ and $\mathcal{T}=\{\Theta_1,...,\Theta_k\}$ a bunch of equivalence relations on $H$ which is closed under $\star$. The latter must be {\em separating} as specified in Def. \ref{separating equivalences}, that is, (i) the closure operator $\sigma_{\Theta_i}$ associated with $\Theta_i$ takes up-sets to up-sets, and (ii)  whenever $x,y\in H$, $\Theta_i\in\mathcal{T}$ and $(x,y)\notin\Theta_i$, then there exists a $\Theta_i$-saturated up-set $V\subseteq H$ containing exactly one of $x,y$. To enhance readability, we abbreviate $\sigma_{\Theta_i}$ by $\sigma_i$ whenever appropriate.

 Our goal is to describe such structures - rather informally -  by sentences of a first-order language $\Lambda_k$ with equality containing a binary relation symbol $\leq$ and a finite number of binary relation symbols $\Theta_1,...,\Theta_k$. It is straightforward how to express by $\Lambda_k$-sentences that $\leq$ is an order relation on $H$ and that the $\Theta_i$ are equivalence relations on $H$. As an example, the sentence $\mathbf{C_{ij}}$ below expresses that $\Theta_i$ and $\Theta_j$ commute:

 $\mathbf{C_{ij}}: \forall xuy\exists u'(x\Theta_i u\Theta_j y\rightarrow\Theta_j u'\Theta_i y)$.

For condition \ref{separating equivalences}(i), remember that any up-set $U\subseteq H$ is a set union of principal up-sets $\uparrow\!x$ with $x\in H$,  so it will do to enforce that $\sigma_i(\uparrow\!x)$ is an up-set for all $x\in H$. Put

$\mathbf{A_i}: \forall xyuv\exists y'((x\leq y\Theta_i u\leq v)\:\rightarrow\: (x\leq y'\Theta_i v))$.

\begin{claim}    \label{A}
$\sigma_i(\uparrow\!x)$ is an up-set iff $(H;\leq,\mathcal{T})$ satisfies $\mathbf{A_i}$ for all $x\in H$ and all $\Theta_i\in\mathcal{T}$.
\end{claim}
\begin{proof}
The formula just says if $u$ is in the $\Theta_i$-class of some $y\in\uparrow\!x$ and $v\geq u$, then $v$ is in the $\Theta_i$-class of some $y'\in\uparrow\!x$, making it a member of $\sigma_i(\uparrow\!x)$.
\end{proof}

For condition \ref{separating equivalences}(ii), observe that $\sigma_i(\uparrow\!x)$ is obviously the least $\Theta_i$-saturated up-set containing $x$, assuming  $\mathbf{A_i}$. Consequently, if $x$ and $y$ can be separated by any $\Theta_i$-saturated up-set, they can be separated by $\sigma_i(\uparrow\!x)$ or $\sigma_i(\uparrow\!y$). So we have to rule out the possibility that simultaneously $x\in\sigma_i(\uparrow\!y)$ and $y\in\sigma_i(\uparrow\!x)$, whenever $(x,y)\notin\Theta_i$.  This is exactly what the following sentence does:

$\mathbf{B_i}: \forall xyx'y'((x\leq x'\Theta_i y  \hspace{0.5em}\&\hspace{0.5em}  y\leq y'\Theta_i x) \rightarrow x\Theta_i y)$.

Summing up,we have

\begin{proposition}   \label{duals of findistrCDFInfAlgs}
The dual objects of finite distributive  information algebras are structures $(H;\leq,\mathcal{T})$ where $H$ is finite, $\leq$ is an order on $H$ and $\mathcal{T}$ is a set of equivalence relations on $H$ satisfying conditions $\mathbf{A_i}$, $\mathbf{B_i}$ and $\mathbf{C_{ij}}$ for all $\Theta_i,\Theta_j\in\mathcal{T}$.
\end{proposition}

This amounts to a first-order description of the dual objects of finite distributive  information algebras. However, $\Lambda_k$ cannot express the property of $\mathcal{T}$ being closed under $\star$. We will address this problem below. The existence of $\star$-closed subsets $\mathcal{T}\subseteq Eq(H)$ consisting of separating equivalences on an arbitrary ordered set $(H;\leq)$ is trivial: Let $\Delta$ be the identity relation on $H$, and $\nabla$ the all-relation $\nabla = H\times H$. It is straightforward to see that both $\Delta$ and $\nabla$ (trivially, since there is nothing to separate) are separating and that $\Delta\star\nabla = \nabla = \nabla\star\Delta$, so the answer is yes. The right question at this place is to ask for {\em nontrivial} such $\mathcal{T}$, meaning $\mathcal{T}\supsetneq\{\Delta,\nabla\}$.

\begin{lemma} \label{nontrivial separating equiv}
On any ordered set $(H,\leq)$ with $|H|\geq 2 $ there exists a nontrivial separating equivalence $\Theta\neq\Delta.\nabla$.
\end{lemma}

\begin{proof}
Pick any up-set $U=\uparrow\!x\neq H$ and define an equivalence $\Theta$ on $H$with the blocks $U$ and all singletons $\{y\}$ with $y\notin\uparrow\!x$. Consider an arbitrary principal up-set $V = \uparrow\!y\subseteq H$. Now if $V\cap U = \emptyset$, then obviously $\sigma_\Theta(V) = V$, and if $V\cap U \neq \emptyset$, then $\sigma_\Theta(V\cap U) = U$ and $\sigma_\Theta(V\setminus U) = V\setminus U$, hence $\sigma_\Theta(V) = V\cup U$, a (not necessarily principal) up-set. So the first part of the separation condition is satisfied.

For the second part, consider $y,y'\in H$ such that $(y,y')\notin\Theta$. If $y\in U$ and $y'\notin U$, then $U$ will do the job. So suppose $y,y'\in H\setminus U$. Since $y\neq y'$ we have $y\nleq y'$ or $y'\nleq y$. Then, borrowing the above argument, $\uparrow\!y'\cup U$ is a $\Theta$-closed up-set containing $y'$ but not $y$.
\end{proof}

\begin{corollary}   \label{nontrivial star-closed subsets}
On any ordered set $(H,\leq)$ with $|H|\geq 2$ there exists a nontrivial $\star$-closed subset $\mathcal{T}\subseteq Eq(H)$ consisting of separating equivalences.
\end{corollary}

\begin{proof}
Take $\mathcal{T} =\{\Theta,\Delta,\nabla\}$ with $\Theta$ as constructed in Lemma \ref{nontrivial separating equiv}.
\end{proof}

Remember that Cor. \ref{preservation of N, A, C, Q} allowed us to restrict our attention to  information algebras where the set of all extraction operations is closed under composition, which corresponds to $\mathcal{T}$ being $\star$-closed. There was a good reason to do so: Otherwise, in Def. \ref{CDF homos}, the $g$-half of a homomorphism $(f,g)$ would become a partial operation which is highly undesirable. So we have a closer look at how $\star$ interacts with separating equivalences. Assume $\Theta_i,\Theta_j\in\mathcal{T}$ are separating.

In order to satisfy Def. \ref{separating equivalences}(i),  the closure operator $\sigma_{\Theta_i\star\Theta_j}$ must  take up-sets to up-sets. Since $\sigma_{\Theta_i\star\Theta_j}= \sigma_i\circ\sigma_j$ by Lemma \ref{commuting equivalences}, this is obvious.

Def. \ref{separating equivalences}(ii) for $\Theta_i\star\Theta_j$ is harder to enforce. We need need a stronger form of $\mathbf{B_i}$ ensuring that whenever $x,y\in H$ and $(x,y)\notin\Theta_i\star\Theta_j$, then there exists a $\Theta_i\star\Theta_j$-saturated up-set $V\subseteq H$ containing exactly one of $x,y$. Now, assuming $\mathbf{A_i}$ and $\mathbf{A_j}$, we have $y\in \sigma_j\sigma_i(\uparrow\!x)$ iff $x\leq x'\Theta_i u\leq u'\Theta_j y$ for some $x',u,u'\in H$, and $x\in \sigma_j\sigma_i(\uparrow\!y)$ iff $y\leq y'\Theta_i v\leq v'\Theta_j x$ for some $y',v,v'\in H$. The following formula rules out the possibility of having $y\in \sigma_j\sigma_i(\uparrow\!x)$ and $x\in \sigma_j\sigma_i(\uparrow\!y)$ simultaneously whenever $(x,y)\notin\Theta_i\star\Theta_j$:

$\mathbf{B_{ij}}: \forall xyx'uu'y'vv'\exists z(((x\leq x'\Theta_i u\leq u'\Theta_j y)\hspace{0.5em}\&\hspace{0.5em}(y\leq y'\Theta_i v\leq v'\Theta_j x))\rightarrow x\Theta_i z\Theta_j y)$.

Note that if $\Theta_i = \Theta_j$, then putting $u=u'$, $v=v'$ and $z=x$ or $z=y$ reduces $\mathbf{B_{ij}}$ to $\mathbf{B_i}$, so $\mathbf{B_{ij}}$ indeed contains $\mathbf{B_i}$.

In our original definition of an information algebra, the set $E$ of extraction operators was not supposed to be closed under composition. Let us call, for convenience, an  algebra  $\underline{A}=(\underline{\Phi};E)$ {\em partial} if this is not necessarily the case. Then the dual structures of finite distributive partial information algebras are exactly the structures $(H;\leq,\mathcal{T})$ where $H$ is finite, $\leq$ is an order on $H$ and $\mathcal{T}$ is a (finite) set of equivalence relations on $H$ satisfying conditions $\mathbf{A_i}$, $\mathbf{B_{ij}}$ and $\mathbf{C_{ij}}$ for all $\Theta_i,\Theta_j\in\mathcal{T}$, but with $\mathcal{T}$ not necessarily closed under $\star$. Write $\mathcal{T}^{\star}$ for the closure of $\mathcal{T}$ under $\star$, then obviously $(H;\leq,\mathcal{T}^{\star})$ will be, by Cor. \ref{preservation of N, A, C, Q}, the dual of an ''ordinary´´ distributive  information algebra. Since obviously $\mathcal{T}^{\star\star} = \mathcal{T}^\star$, we obtain

\begin{proposition}
$(H;\leq,\mathcal{T})$, where $H$ is finite, $\leq$ is an order on $H$ and $\mathcal{T}$ is a (finite) set of equivalence relations on $H$, is the dual structure of a finite distributive  information algebra iff there is subset $\mathcal{G}\subseteq\mathcal{T}$ satisfying $\mathbf{A_i}$, $\mathbf{B_{ij}}$ and $\mathbf{C_{ij}}$ such that $\mathcal{T} = \mathcal{G}^{\star}$.
\end{proposition}

Consider a first order language $\Lambda$ containing $\leq$ and countably many relation symbols $\Theta_1,\Theta_2,\ldots$. We obtain

\begin{theorem}
Then the class of all finite {\em partial} distributive  information algebras is (relatively) $\Lambda$-elementary, and any finite distributive  information arises as the $\star$-closure of a (generally non-unique) member of this class.
\end{theorem}

\section{Summary} \label{Summary}

We considered the intuitive notion of a "piece of information" not by giving a  precise definition of it, but by precisely specifying the rules which should - again intuitively - govern their properties. The basic idea is that "pieces of information" must be able to be combined, and that this combination does not depend on the order in which the pieces under consideration are put into the combination. This is modeled algebraically by a commutative idempotent semigroup $(\Phi,\cdot)$ of "pieces of information" containing a unit 1 which doesn't change any piece of information when combined with it, and a zero 0 (representing contradiction) which outputs 0 when combined with any piece of information.

On other hand, pieces of information are obtained when one asks questions from an abstract set $Q$ of questions, and given a piece of information as an answer, one should be able to extract from this piece the information relevant to the question asked. This defines, for each question, an unary operation from pieces of information to pieces of information. At this point, we make a crucial assumption: We stipulate that, when two question are asked in succession, the information obtained does not depend on order of the questions. This clearly delimits the scope of algebraic theory we developed, but as the literature cited shows, a plethora of important examples falls in this category (see also Subsection \ref{subsec:Expl}). In order to obtain algebras without partial operations, we also stipulate at this point that the set of these operations be closed under under composition. All said and done, we end up with with a second commutative idempotent semigroup $(E,\circ)$ of so-called {\em extraction operators} (on $\Phi$) indexed by $Q$.

Commutative idempotent semigroups may be equipped with a compatible order structure in exactly two ways. To stay in accordance with the existing literature, we opted - for $\Phi$ - for the one making $0$ the greatest and $1$ the least element, which turns $\Phi$ into a bounded join-semilattice. This order was referred to as the {\em information order}. It turned out that extraction operators preserve the information order and that the defining properties of extraction operators may be expressed in order-theoretic terms. This made clear that extraction operators are duals of existential quantifiers as considered in algebraic logic. 

So far, the set $\Phi$ of pieces of information as well as the set $Q$ of questions were arbitrary abstract sets, subject only to the conditions specified for composition and extraction. We proceeded by giving them an internal structure as specific set-theoretical constructs over a (non-empty) base set $U$, best thought of as a set of {\em possible worlds}. Questions $x\in Q$ were then be modelled by equivalence relations $\equiv_x$ on $U$, the idea being that for $u,u'\in U$ we have $u\equiv_x u'$ iff question $x$ has the same answer in the worlds $u$ resp.\!\! $u'$. The point then was to model pieces of information as semilattices of subsets of $U$, and extraction operators as the saturation operators associated with the equivalence relations $\equiv_x$ on $U$ for $x\in Q$. This led to a type of information algebra called {\em set algebra}, for lack of a better term.

The rest of the paper is concerned with representations of abstract information algebras, and with duality theory in the sense of the book "Natural Dualities for the Working Algebraist" by David Clark and Brian Davey, putting information algebras into the context of classical dualities like Stone resp. Priestley duality for Boolean algebras resp. distributive lattices. First, we showed that any information algebra in our sense may be represented by a set algebra as mentioned above. Then, we obtained a direct representation of information algebras containing enough "maximally informative" in terms of Boolean algebras, and finally we showed that the category of information algebras based on a distributive lattice is fully dual - modelling objects as well as morphisms - to a category of certain topological spaces equipped with appropriate equivalence relations.

%\acknowledgements{Acknowledgements}

% Variant A
%==========================================================
% Back Matter (References and Notes)
%----------------------------------------------------------
% Style and layout of the references
%\bibliographystyle{mdpi}
%\makeatletter
%\renewcommand\@biblabel[1]{#1. }
%\makeatother

%\begin{thebibliography}{1}
%%ref 1
%\bibitem{ref-journal}
%Lastname, F.; Author, T. The title of the cited article. {\em Journal Abbreviation} {\bf 2008}, {\em 10}, 142-149.
%
%%ref 2
%\bibitem{ref-book}
%Lastname, F.F.; Author, T. The title of the cited contribution. In {\em The Book Title}; Editor, F., Meditor, A., Eds.; Publishing House: City, Country, 2007; pp. 32-58.
%
%\end{thebibliography}

% Variant B

%==========================================================
%----------------------------------------------------------
% Use the following option to include external BibTeX files:
%\bibliography{template}
%----------------------------------------------------------
%\bibliography{lite}
%\bibliographystyle{mdpi}

\bibliography{tcslit}
\bibliographystyle{authordate3}

\end{document}